\documentclass[12pt]{article}

\usepackage{eurosym}
\usepackage{graphicx}
\usepackage{enumerate}
\usepackage{amsmath}
\usepackage{amsfonts}
\usepackage{amssymb}
\usepackage{amsfonts}
\usepackage{amssymb,amsmath,amsthm,epsfig,graphicx,natbib,subfigure,hyperref}
\usepackage{geometry}
\usepackage{caption}
\usepackage{color}
\usepackage{enumerate}
\usepackage{setspace}
\usepackage[textwidth=30mm]{todonotes}
\usepackage{sgame, tikz}
\usepackage{multirow,array}
\usepackage{rotating}
\usepackage{graphicx}
\usepackage{subfigure}
\usepackage{float}
\usepackage{verbatim}
\usepackage{bbm}
\usepackage{natbib}
\usepackage[ruled,vlined]{algorithm2e}

\definecolor{darkblue}{rgb}{0.0, 0.0, 0.55}
\hypersetup{colorlinks,linkcolor={darkblue},citecolor={darkblue},urlcolor={darkblue}} 
\allowdisplaybreaks

\bibpunct[: ]{(}{)}{;}{a}{}{,}
\newtheorem{theorem}{Theorem}
\newtheorem*{theorem*}{Theorem}
\newtheorem{lemma}{Lemma}

\newtheorem{corollary}{Corollary}

\newtheorem{definition}{Definition}

\newtheorem*{example*}{Example}

\newtheorem{proposition}{Proposition}

\SetKw{KwAnd}{and}

\usepackage{soul}
\usepackage[normalem]{ulem}

\title{Privacy Structure and Blackwell Frontier\thanks{We are grateful for valuable comments and discussions from Yang Sun. All errors are our own.}}
\author{Zhang Xu\thanks{School of Economics, Renmin University of China, China. \textit{\ Email}: \href{mailto:xuzhang@ruc.edu.cn}{xuzhang@ruc.edu.cn}} \and Wei Zhao\thanks{School of Economics and Management, Tsinghua University. \textit{\ Email}: \href{mailto:wei.zhao@outlook.fr}{wei.zhao@outlook.fr}}}

\begin{document}

\maketitle

\begin{abstract}
    This paper characterizes the set of feasible posterior distributions subject to graph-based inferential privacy constraint, including both differential and inferential privacy. This characterization can be done through enumerating all extreme points of the feasible posterior set. A connection between extreme posteriors and strongly connected semi-chains is then established. All these semi-chains can be constructed through successive unfolding operations on semi-chains with two partitions, which can be constructed through classical spanning tree algorithm. A sharper characterization of semi-chains with two partitions for differential privacy is provided. 
\end{abstract}

\newpage
    \tableofcontents
\newpage
\section{Introduction}



Information plays a critical role in decision making and effectiveness of mechanisms. Recent technological development facilitates the collection, storage, acquisition and usage of data. At the same time as efficiency improvement, growing concerns on information leakage have been drawn. Various privacy constraints have been imposed to address these concerns. This paper tries to, subject to exogenous privacy constraint, characterize the feasible information which can be disclosed, through characterizing the feasible posterior distributions. 

The significance of this result is both theoretical and practical. A key insight of Bayesian persuasion (c.f. \citealt{kamenica_bayesian_2011}) is that the set of information structures can be mapped to the set of Bayesian-plausible posterior distributions. Since then, information disclosure is often characterized directly on the space of posterior distributions. In this spirit, theoretically, our work identifies the restrictions that privacy constraints impose on posterior distributions. Practically, characterizing the set of feasible posterior distributions guides optimal, privacy-constrained decision-making. This, in turn, enables the endogenous determination of the optimal degree of privacy protection (c.f. \citealt{abowd_economic_2019}). More generally, the characterization of feasible posterior distributions provides a foundation for privacy-constrained information design (c.f. \citealt{schmutte2022information} and \citealt{pan2025differentially}), which can be derived through concavification over the set of feasible posteriors. The second application is on the characterization of feasible data queries subject to privacy constraints on a sensitive variable (For perfect privacy preservation, see \citealt{strack2024privacy}, \citealt{he2021private}; for relaxed constraints, see \citealt{wang2025inferentially}, \citealt{xu2025privacy}). Our companion paper, \cite{xu2025privacy}, shows that this query characterization problem reduces to characterizing the privacy-constrained posterior distributions on the sensitive variable. This characterization of feasible queries is fundamental to optimal decision-making under such constraints (e.g., price discrimination with perfect privacy preservation; \citealt{strack2025non}).

The most intuitive privacy requirement for signals is based on \emph{pairwise indistinguishability}: after observing a signal, the additional information should not allow the receiver to distinguish a pair of realizations $(\theta_i, \theta_j)$ much more precisely than under the prior. Formally, for almost every posterior belief $\mu \in \Delta(\Theta)$ induced by the signal,
\begin{equation}
    e^{-\varepsilon}\frac{\mu_0(\theta_i)}{\mu_0(\theta_j)} \leq \frac{\mu(\theta_i)}{\mu(\theta_j)} \leq e^{\varepsilon}\frac{\mu_0(\theta_i)}{\mu_0(\theta_j)},
\end{equation}
where $\mu_0$ is the prior belief and $\varepsilon \ge 0$ controls the allowable degree of distinguishability. This condition provides a \emph{linear} measure of how much additional information about the pair $(\theta_i, \theta_j)$ the signal may reveal. When this requirement is imposed on all pairs of realizations, it induces $\varepsilon$-inferential privacy (Definition~\ref{def:inferential}). When it is imposed only on neighboring realizations, it induces $\varepsilon$-differential privacy (Definition~\ref{def:differential}). More generally, the pairwise indistinguishability structure can be represented by a graph on $\Theta$, where an edge between two realizations indicates that the information disclosed about this pair must satisfy the privacy requirement. We call this definition of privacy as graph-based inferential privacy.

In this paper, we focus on graph-based inferential privacy. Our objective is to characterize the Blackwell frontier of privacy-preserving signals. Lemma~\ref{lemma:extreme_points} establishes that a privacy-preserving signal is Blackwell-undominated if and only if almost every possible posterior belief is an extreme point. Therefore, to characterize the Blackwell frontier, it suffices to identify the extreme points of the feasible belief set.

Our first result establishes a connection between extreme points and semi-chains of the underlying graph (Theorem~\ref{thm:semi-chain}). An $L$-semi-chain is a specific structure imposed on a graph that partitions its nodes into $L$ levels such that every edge connects either two nodes within the same level or nodes in two adjacent levels. Given an $L$-semi-chain, the posterior assigns to each realization in $l$-level  a probability that is proportional to $l$ times its prior probability.

When the graph is complete, corresponding to  $\varepsilon$-inferential privacy, the definition of a semi-chain implies that only 2-semi-chains exist, and any such partition corresponds to an extreme posterior belief (Corollary~\ref{cor:inferential}). In contrast, for a general graph, the question arises: how can we systematically generate all of its semi-chains? Our second main result shows that any $L$-semi-chain can be generated from a 2-semi-chain through \emph{unfolding} (Theorem~\ref{thm:2-semi-chain}). Unlike $L$-semi-chains with $L \geq 3$, 2-semi-chains have only two levels, so we do not need to enforce the constraint that edges cannot exist between non-adjacent levels. This simplification allows us to generate all 2-semi-chains by enumerating the spanning trees of the graph.

Under the setting of $\varepsilon$-differential privacy, where each realization is a $K$-dimensional vector, $L$-semi-chains exist if and only if $L \leq K+1$ (Proposition~\ref{prop:constrain_level_dp}). In the binary case, where each dimension takes values in $\{0,1
\}$, there exists a unique 2-semi-chain up to reversing the first and second levels (Proposition~\ref{prp:dp_binary_unique}). For the two-dimensional case, we provide an algorithm that generates all 2-semi-chains without duplication (Proposition~\ref{prp:dp_2dim}).


\section{Graph-Based Inferential Privacy}

\emph{Database.} A sensitive database is a random variable $\theta$, distributed according to an interior distribution $\mu_0$ over a finite set $\Theta := \{\theta_1, \ldots, \theta_J\}$. In practice, $\theta$ may be multidimensional and capture all relevant characteristics of an agent or a set of agents, such as history records, income, age, gender, race, address, etc.

\emph{Signals.} A signal is a mapping $\pi: \Theta \to \Delta(S)$, where $S \subseteq \mathbb{R}$ is the signal space. Observing a signal realization $s \in S$ induces a posterior $\mu_s \in \Delta(\Theta)$. The signal $\pi$ thus induces a distribution over posteriors, denoted by $\langle \pi \rangle$.


\emph{Graph-Based Inferential Privacy.} Let $\mathbf{G} = (g_{ij})_{i,j=1}^J \in \{0,1\}^{J \times J}$ denote an undirected, connected graph over the set of nodes $\Theta$. For any pair $(\theta_i, \theta_j) \in \Theta^2$, we set $g_{ij} = 1$ if there is a link between $\theta_i$ and $\theta_j$, and $g_{ij} = 0$ otherwise. The graph is symmetric, i.e., $g_{ij} = g_{ji}$ for all $i,j \in \{1,\ldots,J\}$, and, without loss of generality, $g_{ii} = 0$ for all $i$.
With a slight abuse of notation, we write $(\theta_i, \theta_j) \in \mathbf{G}$ whenever $g_{ij} = 1$.

Given a $\mathbf{G}$ and some $\varepsilon \geq 0$, we define $$\mathcal{M}_{\mathbf{G}}^\varepsilon :=\left\{\mu \in \Delta(\Theta) :  \frac{\mu(\theta_{i})}{\mu_0(\theta_{i})} \leq e^{\varepsilon}\frac{\mu(\theta_{j})}{\mu_0(\theta_{j})} \text{ for all } (\theta_i, \theta_j ) \in \mathbf{G}\right\},$$
which is a closed convex subset of $\Delta(\Theta)$. Let $\operatorname{ext}\mathcal{M}_{\mathbf{G}}^\varepsilon$ as the set of extreme points of $\mathcal{M}_{\mathbf{G}}^\varepsilon$, i.e., $\operatorname{ext}\mathcal{M}_{\mathbf{G}}^\varepsilon := \left\{ \mu \in \mathcal{M}_{\mathbf{G}}^\varepsilon : \nexists \mu' \neq \mu'' \in \mathcal{M}_{\mathbf{G}}^\varepsilon, \, \alpha \in (0,1) \text{ such that } \mu = \alpha \mu' + (1-\alpha) \mu'' \right\}$. For any $\mu \in \operatorname{ext} \mathcal{M}_{\mathbf{G}}^\varepsilon$, we refer to it as an extreme posterior (with respect to $\mathcal{M}_{\mathbf{G}}^\varepsilon$).

\begin{definition}\label{def:inferential}
    A signal $\pi$ is $(\mathbf{G},\varepsilon)$-inferential privacy-preserving if $\left<\pi\right>(\mathcal{M}_{\mathbf{G}}^\varepsilon) = 1$.
\end{definition}


Since $\mathbf{G}$ is undirected and  symmetric, for any $(\theta_i, \theta_j) \in \mathbf{G}$, if $\mu \in \mathcal{M}_{\mathbf{G}}^\varepsilon$, then 
\begin{equation*}
    e^{-\varepsilon} \frac{\mu(\theta_{j})}{\mu_0(\theta_{j})} \leq \frac{\mu(\theta_{i})}{\mu_0(\theta_{i})} \leq e^{\varepsilon}\frac{\mu(\theta_{j})}{\mu_0(\theta_{j})}.
\end{equation*}
When $\varepsilon = 0$, if the database is either $\theta_i$ or $\theta_j$, then the receiver’s posterior belief over $\theta_i$ and $\theta_j$ coincides with the prior, i.e. \begin{equation*}
    \frac{\mu(\theta)}{\mu(\theta_i) + \mu(\theta_j)} = \frac{\mu_0(\theta)}{\mu_0(\theta_i) + \mu_0(\theta_j)}, \quad \forall \theta \in\{\theta_i,\theta_j\}.
\end{equation*} In other words, the receiver cannot make any additional inference from the signal to distinguish between $\theta_i$ and $\theta_j$. When $\varepsilon > 0$ but sufficiently small, this inference remains limited. The graph $\mathbf{G}$ specifies which pairs of information are not allowed to be inferred from the signal.


Graph-based inferential privacy is developed within the framework introduced by \citet{kifer2014pufferfish}. In their framework, one may specify a subset of prior distributions, and the privacy constraints are required to hold only with respect to those priors. In contrast, we focus on a fixed information structure that is independent of the prior belief. Consequently, if the privacy constraints are satisfied under any interior prior, they are satisfied under all priors.

One important special case of graph-based inferential privacy is \emph{differential privacy}, originally proposed by \citet{dwork2006calibrating}. This notion has achieved remarkable success: beyond an extensive and growing academic literature, it has been adopted as a core privacy framework by major institutions such as Google, Microsoft, and the U.S. Census Bureau \citep{abowd2018us}. 
Under differential privacy, only information differences between \emph{neighboring} realizations of the dataset are protected. Specifically, let $\Theta = \prod_{k=1}^K \Theta^{(k)}$, where each $\Theta^{(k)}$ is finite and $\theta = (\theta^{(1)}, \ldots, \theta^{(K)})$. 
Two realizations $\theta_i, \theta_j \in \Theta$ are said to be neighboring if and only if there exists $\hat{k} \in \{1, \ldots, K\}$ such that $$\theta_i^{(\hat{k})} \neq \theta_j^{(\hat{k})}, \quad  \theta_i^{(-\hat{k})} = \theta_j^{(-\hat{k})},$$ where $\theta^{(-k)} := (\theta^{(1)}, \ldots, \theta^{(k-1)}, \theta^{(k+1)}, \ldots, \theta^{(K)}).$ 
We can then define $\mathbf{G}_{\mathrm{diff}}$ as the \emph{differential graph}, where $(\theta_i, \theta_j) \in \mathbf{G}_{\mathrm{diff}}$ if and only if $\theta_i$ and $\theta_j$ are neighboring.

\begin{definition}\label{def:differential}
    A signal $\pi$ is $\varepsilon$-differential privacy if it is $(\mathbf{G}_{\mathrm{diff}},\varepsilon)$-inferential privacy.
\end{definition}

Another special case arises when $\mathbf{G}$ is a complete graph, denoted by $\mathbf{G}^*$, where $(\theta_i,\theta_j) \in \mathbf{G}^*$ for all $i \neq j$. This formulation was introduced by \citet{ghosh2016inferential} and subsequently applied by \citet{wang2025inferentially} to study privacy-preserving information disclosure.

\begin{definition}
    A signal $\pi$ is $\varepsilon$-inferential-preserving if it is $(\mathbf{G}^*, \varepsilon)$-inferential privacy, where $\mathbf{G}^*$ is the complete graph on $\Theta$.
\end{definition}

\subsection{Blackwell Frontier and Extreme Posteriors}

\begin{definition}[Blackwell Dominance]\label{def:blackwell_dominace}
    A signal $\pi$ is Blackwell-dominated by a signal $\pi'$, denoted by $\pi \preceq \pi'$, if for any action set $A$ and measurable function $u : \Theta \times A \to \mathbb{R}$,
    \begin{equation*}
        \mathbb{E}_{s\sim \pi}\left[\sup_{a \in A} \mathbb{E}[u(\theta,a)|s]\right] \leq \mathbb{E}_{s\sim \pi'}\left[\sup_{a \in A} \mathbb{E}[u(\theta,a)|s]\right].
    \end{equation*}
    A signal $\pi$ is Blackwell-equivalent to $\pi'$, written $\pi \sim \pi'$, if $\pi \preceq \pi'$ and $\pi' \preceq \pi$. It is strictly Blackwell-dominated, written $\pi \prec \pi'$, if $\pi \preceq \pi'$ but not $\pi' \preceq \pi$.
\end{definition}

Denote $\Pi_{\mathcal{M}_{\mathbf{G}}^\varepsilon}$ as the set of $({\mathbf{G}},\varepsilon)$-inferential privacy-preserving signals and $\overline{\Pi}_{\mathcal{M}_{\mathbf{G}}^\varepsilon}$ is its Blackwell frontier, i.e.,
$$\overline{\Pi}_{\mathcal{M}_{\mathbf{G}}^\varepsilon} := \bigl\{ \pi \in \Pi_{\mathcal{M}_{\mathbf{G}}^\varepsilon} : \nexists \, \pi' \in \Pi_{\mathcal{M}_{\mathbf{G}}^\varepsilon} \text{ such that } \pi \prec \pi' \bigr\}.$$

\begin{lemma}\label{lemma:extreme_points}
\begin{enumerate}
    \item A signal is $(\mathbf{G},\varepsilon)$-inferential privacy-preserving if and only if it is Blackwell-dominated by some signal in $\overline{\Pi}_{\mathcal{M}_{\mathbf{G}}^\varepsilon}$.
    \item A signal $\pi \in \overline{\Pi}_{\mathcal{M}_{\mathbf{G}}^\varepsilon}$ if and only if $\left<\pi\right>(\operatorname{ext}\mathcal{M}_{\mathbf{G}}^\varepsilon) = 1$.
\end{enumerate}
\end{lemma}

According to Lemma~\ref{lemma:extreme_points}: (i) to characterize all $(\mathbf{G},\varepsilon)$-inferentially private signals, it suffices to describe their Blackwell frontier, since any other privacy-preserving signal can be obtained by garbling \citep{blackwell1953equivalent}; (ii) to characterize the Blackwell frontier, it is enough to identify all extreme posteriors, as $\overline{\Pi}_{\mathcal{M}_\mathbf{G}^\varepsilon}$ can be generated by convex combinations of extreme posteriors with respect to $\mu_0$.


\subsection{Notation}

To streamline the subsequent analysis, we introduce several graph-theoretic concepts that will be used throughout the paper.

\begin{definition}[Spanning Tree]
    A graph $\mathbf{T}$ is said to be a spanning tree of $\mathbf{G}$ if it is a connected subgraph of $\mathbf{G}$ with no cycles.
\end{definition}

Note that a spanning tree always has $J-1$ edges. 

\begin{definition}[Ordered Partition]
    $\{\Theta_l\}_{l=1}^L$ is a $L$-partition of $\Theta$ if $\Theta_l \neq \emptyset$ for all $l \in \{1, \ldots, L\}$, $\Theta_l \cap \Theta_{l'} = \emptyset$ for any $l \neq l'$ and $\cup_{l=1}^L \Theta_l = \Theta$. A partition $\{\Theta_l\}_{l=1}^L$ with an order such that $\Theta_l$ is strictly lower than $\Theta_{l+1}$ for any $l \in \{1, \dots, L-1\}$, is called an ordered partition, denoted by $(\Theta_l)_{l=1}^L$.
\end{definition}

\begin{definition}[Semi-Chain]
    An ordered partition $C^L = (\Theta_l)_{l=1}^L$ with $L \geq 2$ is called an $L$-semi-chain of $\mathbf{G}$, if for any $\theta \in \Theta$, if $\theta \in \Theta_l$, then $N(\theta) \subseteq \Theta_{l-1} \cup \Theta_{l} \cup \Theta_{l+1}$, where $$N(\theta) := \{\theta' \in \Theta: (\theta,\theta') \in \mathbf{G}\}$$
    is the neighboring set of $\theta$ and $\Theta_0, \Theta_{L+1} := \emptyset$. The $l$-th element $\Theta_l$ is referred to as the $l$-level of $C^L$. Given a semi-chain, an edge $(\theta_{i}, \theta_{j}) \in \mathbf{G}$ is called a within-level edge if $\theta_{i}$ and $\theta_{j}$ belong to the same level, and a between-level edge otherwise.

    A semi-chain $C$ is strongly connected if, after deleting all within-level edges, the remaining graph of $\mathbf{G}$ is still connected.
\end{definition}

In \cite{concas2021chained}, a graph is said to be $L$-semi-chained if there exists an $L$-semi-chain based on this graph, and $L$-chained if there exists an $L$-chain, that is, an $L$-semi-chain without within-level edges.

\section{Extreme Points of Feasible Posteriors}

In this section, we establish the connection between the extreme points of $\mathcal{M}_{\mathbf{G}}^\varepsilon$ and the strongly connected semi-chains of the graph $\mathbf{G}$. The later can be enumerated by leveraging the spanning trees of the graph $\mathbf{G}$, for which efficient enumeration algorithms are well established in the graph theory literature. To this end, we first associate each $\mu \in \mathcal{M}_{\mathbf{G}}^\varepsilon$ with a weighted directed graph $\mathbf{W} = (w_{ij})_{i,j=1}^J \in \mathbb{R}^{J\times J}$. This representation allows us to reformulate the linear programming problem in terms of $\mathbf{W}$.

Given a $\mu \in \mathcal{M}_\mathbf{G}^\varepsilon$, for any pair $(\theta_{i}, \theta_{j}) \in \Theta^2$, there is a $w_{ij}$ such that
\begin{equation}\label{eq:constraints}
    \delta_{j} = \delta_{i} + w_{ij } \varepsilon,
\end{equation}
where $\delta_j := \ln \mu(\theta_j) - \ln \mu_0(\theta_j)$ for all $j \in \{1, \ldots, J\}$.\footnote{Since $\mathbf{G}$ is connected, when $\varepsilon >0$, for any $\mu \in \mathcal{M}_{\mathbf{G}}^\varepsilon$ and $\theta \in \Theta$, $\mu(\theta) > 0$.} Then, the constraints defining $\mathcal{M}_\mathbf{G}^\varepsilon$ can be expressed as restrictions on $\mathbf{W} = (w_{ij})_{i,j=1}^J \in \mathbb{R}^{J\times J}$, i.e.,
\begin{equation}\label{eq:con_1}
    w_{ij} \in [-1,1] \quad \text{ for all } (\theta_{i}, \theta_{j}) \in \mathbf{G}.
\end{equation}

Given a $\mu \in \Delta(\Theta)$, there is a unique corresponding $\mathbf{W}$; conversely, given a $\mathbf{W}$, there is at most one corresponding $\mu$. However, not every $\mathbf{W}$ can induce a distribution.

\begin{lemma}\label{lemma:feasible_w}
    A matrix $\mathbf{W}$ can be obtained by $\mu \in \mathcal{M}_\mathbf{G}^\varepsilon$ if and only if 
    it satisfies \eqref{eq:con_1} and
    \begin{align}
    w_{ij}  = - w_{ji} & \quad   \text{ for any } i,j \in \{1,\ldots, J\}, \label{eq:con_2} \\
    w_{j_1j_r}  = \sum_{r' =1}^{r-1} w_{j_{r'} j_{r'+1}} & \quad  \text{ for any } j_1, j_2, \ldots, j_r \in \{1,\ldots,J\}.\label{eq:con_3}
\end{align}
\end{lemma}

Here, equations \eqref{eq:con_2} and \eqref{eq:con_3} characterize the interdependence among the entries. Since a $\mu \in \Delta(\Theta)$ only have $(J-1)$ freedom, once the values of a sequence $w_{j_1 j_2}, \ldots, w_{j_{J-1} j_J}$ are fixed, all other entries of $\mathbf{W}$ are endogenously determined by \eqref{eq:con_2} and \eqref{eq:con_3}.\footnote{By permuting $\mathbf{W}$, specifying the values of  $w_{j_1 j_2}, w_{j_2 j_3}, \ldots, w_{j_{J-1} j_J}$  is equivalent to specifying the entries of the first superdiagonal $(w_{j,j+1})$.  The second superdiagonal $(w_{j,j+2})$ can then be determined from $(w_{j,j+1})$ via~\eqref{eq:con_3}, and the third superdiagonal can be determined from the second, and so on. By~\eqref{eq:con_2} the diagonal entries are equal to zero and the lower triangular part can be determined by the upper counterpart.} By the correspondence between basic feasible solutions and extreme points (Theorem 3.1, pp.102, \cite{bazaraa2011linear}), each extreme point of $\mathcal{M}_\mathbf{G}^\varepsilon$ is uniquely characterized by $(J-1)$ independent entries $w_{ij}$ satisfying $|w_{ij}| = 1$. This assertion can be expressed as Lemma \ref{lemma:connecteed_acyclic}.

\begin{lemma}\label{lemma:connecteed_acyclic}
    Suppose $\mathbf{W}$ corresponds to some $\mu\in \mathcal{M}_\mathbf{G}^\varepsilon$. $\mu \in \operatorname{ext} \mathcal{M}_\mathbf{G}^\varepsilon$ if and only if there exists a spanning tree $\mathbf{T}$ of $\mathbf{G}$ such that $w_{ij} =\pm 1$ for every $(\theta_{i}, \theta_{j}) \in \mathbf{T}$.
\end{lemma}

The prerequisite of Lemma~\ref{lemma:connecteed_acyclic} is that the matrix $\mathbf{W}$ is induced by some $\mu \in \mathcal{M}_{\mathbf{G}}^\varepsilon$. Once this holds, Lemma~\ref{lemma:connecteed_acyclic} allows us to verify whether $\mathbf{W}$ corresponds to an extreme posterior. However, note that for a given spanning tree $\mathbf{T}$ of $\mathbf{G}$, not every arbitrary assignment of $\pm 1$ weights to its edges yields a valid matrix $\mathbf{W}$ via \eqref{eq:con_2} and \eqref{eq:con_3}, as the resulting $\mathbf{W}$ may violate condition~\eqref{eq:con_1}. Figure~\ref{fig:weighted_spanning_tree} illustrates this point.

\begin{figure}[t]
    \centering
    \includegraphics[width=0.9\linewidth]{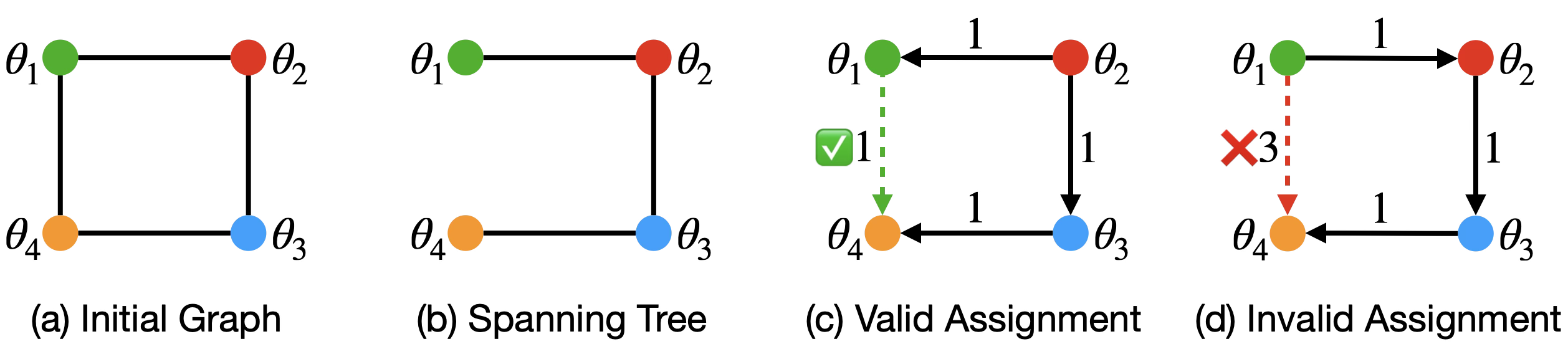}
    \caption{Illustration of Weight Assignments on a Spanning Tree}
    \label{fig:weighted_spanning_tree}

    \vspace{2mm}
    {\small \textit{Note:} Panels (c) and (d) illustrate two different weight assignments on the spanning tree shown in Panel (b). In both panels, the notation “$\theta_i \overset{1}{\rightarrow} \theta_j$” indicates that $w_{ij} = 1$. In Panel (c), the assignment induces $w_{14} = 1$, which satisfies condition~\eqref{eq:con_1}, whereas in Panel (d), $w_{14} = 3 > 1$ violates condition~\eqref{eq:con_1}.}
\end{figure}

To construct all valid weight assignments on a spanning tree, we now shift our focus to the concept of semi-chains.

\subsection{Connection Between Extreme Posteriors and Semi-Chains}

\begin{theorem}\label{thm:semi-chain}
    A posterior $\mu \in \operatorname{ext} \mathcal{M}_\mathbf{G}^\varepsilon$ if and only if there exists a strongly connected $L$-semi-chain $C=(\Theta_l)_{l =1}^L$ such that the corresponding matrix $\mathbf{W}$ of $\mu$ satisfies
    \begin{equation*}
        w_{i j} =
        \begin{cases}
        1 & \text{if there exists } l \in \{1,\ldots, L-1\} \text{ such that } \theta_{i} \in \Theta_l, \theta_{j} \in \Theta_{l+1} \\
        0 & \text{if there exists } l \in\{1,\ldots,L\} \text{ such that } \theta_{i}, \theta_{j} \in \Theta_l
        \end{cases}.
    \end{equation*}
\end{theorem}

In Theorem~\ref{thm:semi-chain}, given an $L$-semi-chain $C$, the corresponding matrix $\mathbf{W}$ is determined endogenously:
the weight of any within-level edge equals $0$, and the absolute value of the weight of any between-level edge equals $1$, with direction oriented from the lower to the higher level.
All remaining entries of $\mathbf{W}$ are then uniquely determined by the condition~\eqref{eq:con_3}.
In what follows, unless stated otherwise, we say that a posterior $\mu$ is generated by $C$, implicitly referring to the $\mathbf{W}$ defined as in Theorem~\ref{thm:semi-chain}.

Since there are no edges between nodes in non-adjacent levels of a semi-chain, the corresponding $\mathbf{W}$ automatically satisfies condition~\eqref{eq:con_1}. Moreover, every strongly connected semi-chain contains a spanning tree whose edges have weights $\pm 1$. Therefore, a semi-chain naturally determines a valid weight assignment on a spanning tree.

Theorem~\ref{thm:semi-chain} establishes a one-to-one correspondence between the extreme points of $\operatorname{ext} \mathcal{M}_{\mathbf{G}}^\varepsilon$ and the set of all strongly connected $L$-semi-chains.
Given an $L$-semi-chain $C = (\Theta_l)_{l=1}^L$, the corresponding posterior distribution $\mu$ is generated by
\begin{equation*}
    \mu(\theta_i) =
\frac{e^{l\varepsilon}\mu_0(\theta_i)}
{\sum_{l'=1}^L e^{l'\varepsilon}\mu_0(\Theta_{l'})},
\quad \text{if } \theta_i \in \Theta_l.
\end{equation*}
That is, for any $\theta \in \Theta_l$, the posterior probability is proportional to $e^{l\varepsilon}\mu_0(\theta)$, up to normalization. For example, suppose $C = (\Theta_l)_{l=1}^3$ with $\Theta_1 = \{\theta_1\}$, $\Theta_2 = \{\theta_2\}$, and $\Theta_3 = \{\theta_3\}$ is a 3-semi-chain of some graph. Then the corresponding posterior $\mu$ is given by
\begin{align*}
    \mu(\theta_1) &= \frac{\mu_0(\theta_1)}{\mu_0(\theta_1) + e^\varepsilon \mu_0(\theta_2) + e^{2\varepsilon}\mu_0(\theta_3)}, \\
    \mu(\theta_2) &= \frac{e^\varepsilon \mu_0(\theta_2)}{\mu_0(\theta_1) + e^\varepsilon \mu_0(\theta_2) + e^{2\varepsilon}\mu_0(\theta_3)}, \\
    \mu(\theta_3) &= \frac{e^{2\varepsilon} \mu_0(\theta_3)}{\mu_0(\theta_1) + e^\varepsilon \mu_0(\theta_2) + e^{2\varepsilon}\mu_0(\theta_3)}.
\end{align*}

When we consider $\varepsilon$-inferential privacy, since $\mathbf{G}^*$ is a complete graph, only 2-semi-chains exist. Moreover, any partition of $\Theta$ can generate a pair of strongly connected 2-semi-chains.

\begin{corollary}\label{cor:inferential}
    A signal $\pi$ is $\varepsilon$-inferential-privacy-preserving if and only if for almost every $\mu \in \operatorname{supp}(\left<\pi\right>)$, there exists a partition of $\Theta$, $\{\Theta_1, \Theta_2\}$ such that
    \begin{equation*}
        \mu(\theta) = \left\{
        \begin{aligned}
            \frac{e^\varepsilon\mu(\theta)}{e^\varepsilon\mu(\Theta_1) + \mu(\Theta_2)} \quad \text{ if } \theta \in \Theta_1 \\
            \frac{\mu(\theta)}{e^\varepsilon\mu(\Theta_1) + \mu(\Theta_2)} \quad \text{ if } \theta \in \Theta_2
        \end{aligned}.\right.
    \end{equation*}
\end{corollary}

Corollary~\ref{cor:inferential} provides an explicit characterization of all extreme posteriors in $\varepsilon$-inferential privacy. \cite{xu2025privacy} further shows that Corollary~\ref{cor:inferential} continues to hold even when $\Theta$ is an infinite set.

\subsection{Connection Between Semi-Chains and Spanning Trees}

Translating extreme posteriors into strongly connected semi-chains is only the first step; we must also establish how to systematically enumerate all strongly connected semi-chains for a given graph. 

Firstly, for a strongly connected 2-semi-chain, it can be constructed from a spanning tree of the initial graph $\mathbf{G}$. Given such a spanning tree, if one node is assigned to a particular level, then all its neighboring nodes are placed on the opposite level. Since a spanning tree contains no cycles, this assignment is well-defined and uniquely determines a strongly connected 2-semi-chain. Moreover, each spanning tree corresponds to two strongly connected 2-semi-chains, which are inverses of each other. Figure~\ref{fig:spanning_tree} demonstrates this relationship. Therefore, to generate all possible strongly connected 2-semi-chains, it suffices to enumerate all spanning trees of $\mathbf{G}$. Efficient algorithms for generating spanning trees are well established in the literature.

However, for any strongly connected $L$-semi-chain with $L \ge 3$, it cannot be directly generated from a spanning tree. The difficulty lies in satisfying the constraint that no edges exist between nodes belonging to non-adjacent (or identical) levels, which is precisely the issue illustrated in Figure~\ref{fig:weighted_spanning_tree}.

To resolve this issue, we first introduce the operation \emph{unfolding} and then demonstrate that any strongly connected $L$-semi-chain with $L \geq 3$ can be generated from a strongly connected 2-semi-chain through successive applications of unfolding. 

\begin{figure}
    \centering
    \includegraphics[width=0.8\linewidth]{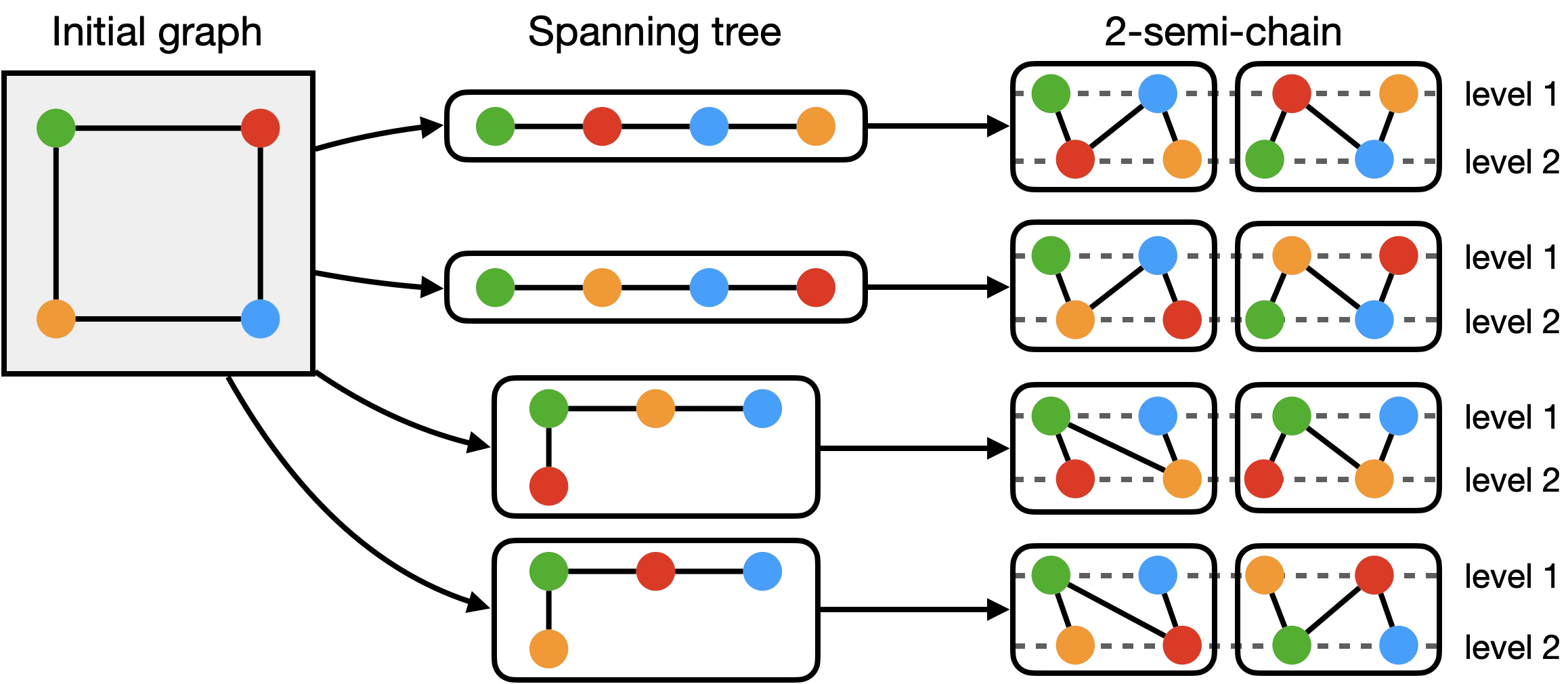}
    \caption{Generating Strongly Connected 2-Semi-Chains Through Spanning Trees}
    \label{fig:spanning_tree}
\end{figure}

\begin{definition}
    Let $C^L = (\Theta_l)_{l=1}^L$ be an $L$-semi-chain for some $L \geq 2$.
    \begin{enumerate}
        \item An ordered partition $(\Theta_{2}^{1}, \Theta_{1}, \Theta_{2}^{2}, (\Theta_l)_{l=3}^L)$ is an upward unfolding of $C^L$ if $\{\Theta_2^1, \Theta_2^2\}$ is a partition of $\Theta_2$ such that $$(\theta_{i}, \theta_{j}) \notin \mathbf{G}, \quad \text{ for all } \theta_{i} \in \Theta_2^1 \text{ and } \theta_{j} \in \Theta_2^2 \cup \Theta_3.$$

        \item An ordered partition $((\Theta_l)_{l=1}^{L-2}, \Theta_{L-1}^{1}, \Theta_{L}, \Theta_{L-1}^{2})$ is a downward unfolding of $C^L$ if $\{\Theta_{L-1}^{1}, \Theta_{L-1}^{2}\}$ is a partition  of $\Theta_{L-1}$ such that 
        $$(\theta_{i}, \theta_{j}) \notin \mathbf{G}, \quad \text{ for all } \theta_{i} \in \Theta_{L-1}^2 \text{ and } \theta_{j} \in \Theta_{L-1}^{1} \cup \Theta_L.$$
        
        \item An $L_+$-semi-chain $C^{L_+}$, for some $L_+ > L$, is obtained from $C^L$ by successive upward (resp. downward) unfolding if there exists a sequence 
        $C^L, C^{L+1}, \ldots, C^{L_+}$
        such that each $C^{L'}$ is an upward (resp. downward) unfolding of $C^{L'-1}$, for all $L' \in \{L+1, \ldots, L_+\}$.
    \end{enumerate}
\end{definition}

\begin{theorem}\label{thm:2-semi-chain}
    An ordered partition $(\Theta_l)_{l=1}^L$ for some $L \geq 3$ is a (strongly connected) $L$-semi-chain if and only if it can be obtained from a (strongly connected) $2$-semi-chain by successive upward unfolding (and symmetrically, by successive downward unfolding).
\end{theorem}

\begin{figure}
    \centering
    \includegraphics[width=0.9\linewidth]{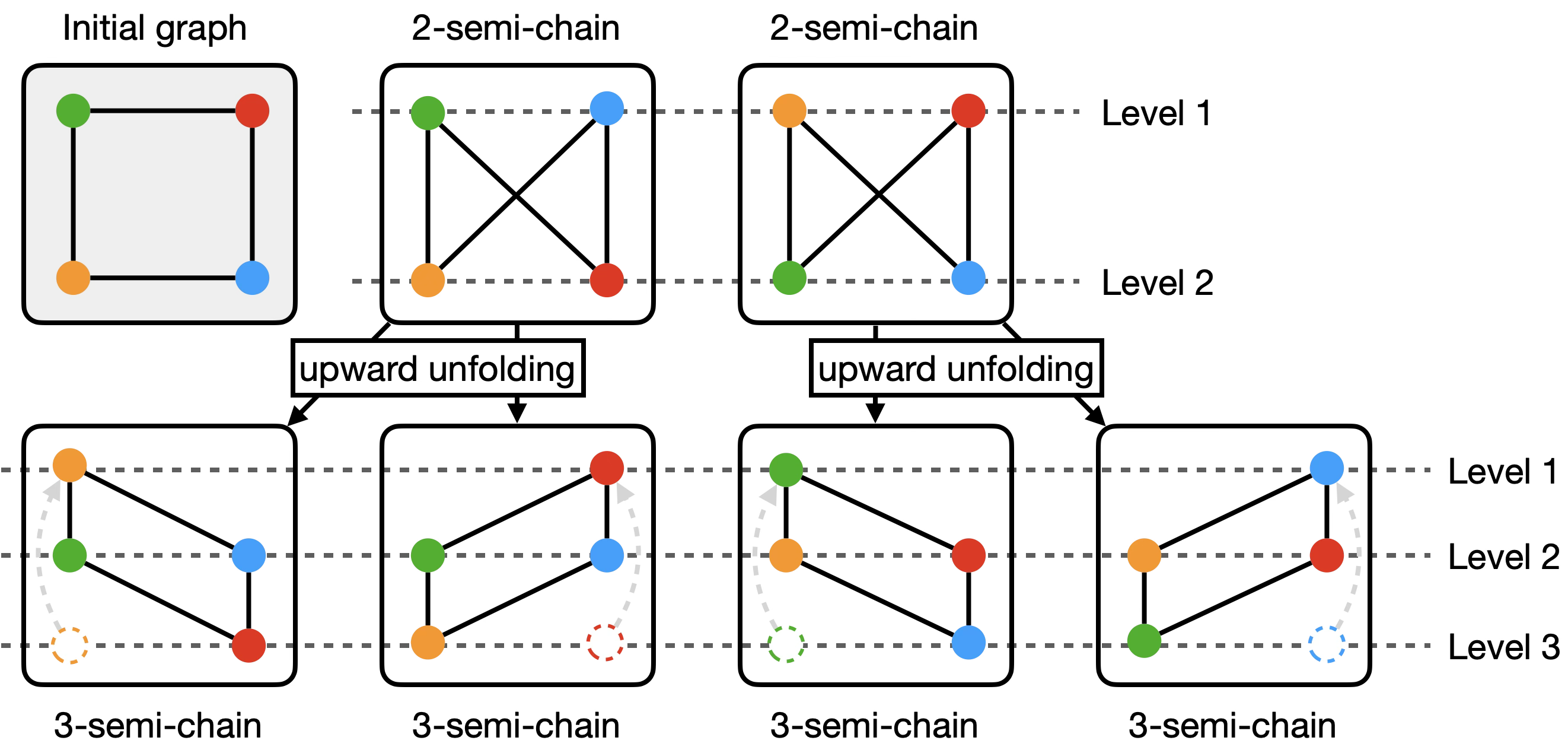}
    \caption{Illustration of Upward Unfolding}
    \label{fig:upward_folding}
\end{figure}

In Figure~\ref{fig:upward_folding}, we give an example to show how to generate all semi-chain trough upward unfolding of 2-semi-chain. Consequently, all strongly connected semi-chains can be obtained by enumerating the spanning trees and applying the possible unfolding operations. However, directly generating 2-semi-chains through spanning trees may lead to duplication. As illustrated in Figure~\ref{fig:spanning_tree}, different spanning trees can produce the same pair of 2-semi-chains. In the next section, we present more efficient methods to characterize 2-semi-chains within differential privacy.

\section{Differential Privacy}

In this section, we focus on differential privacy (Definition \ref{def:differential}). When $\Theta = \prod_{k=1}^K \Theta^{(k)}$ with $|\Theta^{(k)}| \geq 2$, we refer to it as the $K$-dimensional case. We first show that for any $K$-dimensional differential privacy setting, regardless of the number of possible values in each dimension, the number of levels in a semi-chain depends only on the dimensionality $K$.

\begin{proposition}\label{prop:constrain_level_dp}
     For the differential graph $\mathbf{G}_\mathrm{diff}$ in the $K$-dimensional case, a (strongly connected) $L$-semi-chain exists if and only if $L \leq K+1$.
\end{proposition}

Proposition~\ref{prop:constrain_level_dp} consists of two parts. First, the number of levels in a semi-chain cannot exceed $K+1$. This is because, in an $L$-semi-chain, the shortest path connecting a node in the first level to a node in the last level has length at least $L-1$, while in a $K$-dimensional differential graph, the shortest path between any two nodes has length at most $K$. Second, a $(K+1)$-semi-chain does exist. A $K$-dimensional differential graph can be viewed as a combination of $|\Theta^{(k)}|$ copies of $(K-1)$-dimensional differential graphs. Suppose $\hat{C}^{K}$ is the $K$-semi-chain in the $(K-1)$-dimensional case. By duplicating $\hat{C}^{K}$ $|\Theta^{(k)}|$ times and shifting one copy upward by one level, we obtain a $(K+1)$-semi-chain $\hat{C}^{(K+1)}$ for the $K$-dimensional graph. Figure~\ref{fig:dp_construction_m+1} illustrates this construction. Moreover, if $\hat{C}^{K}$ is strongly connected, then $\hat{C}^{(K+1)}$ remains strongly connected.

\begin{figure}
    \centering
    \includegraphics[width=\linewidth]{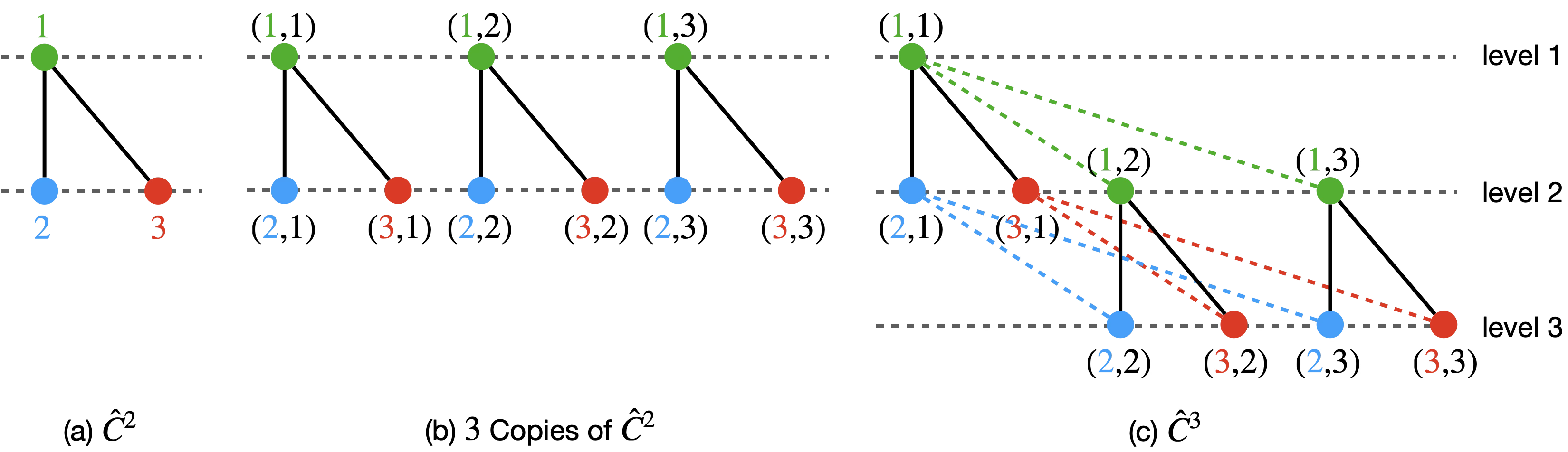}
    \caption{Illustration of Construction of $\hat{C}^{3}$ from $\hat{C}^{2}$}
    \label{fig:dp_construction_m+1}

    \vspace{2mm}
    {\small \textit{Note:} In each panel, only the between-level edges are displayed. Panel (a) presents a 2-semi-chain $\hat{C}^2$ in the two-dimensional case where $\Theta = \{1,2,3\}$. To construct a 3-semi-chain in the three-dimensional case where $\Theta = \{1,2,3\}^2$, Panel (b) duplicates $\hat{C}^2$ three times, and Panel (c) shifts the first copy upward by one level. Since each $\hat{C}^2$ is connected, and after the shift, the nodes with the same value in the last dimension are linked with one another (shown as dashed lines in Panel (c)), the resulting $\hat{C}^3$ is strongly connected. }
\end{figure}

Proposition~\ref{prop:constrain_level_dp} is particularly useful when the dataset has low dimensionality, since in that case we do not need to consider overly complex semi-chains. In the two-dimensional case, regardless of how many possible values each dimension takes, only 2-semi-chains exist. Therefore, to generate all extreme posteriors, the unfolding operation is unnecessary. However, when the dimensionality is high, even if each dimension contains only two possible values, more complex semi-chains may arise. In the next subsection, we will show that in this case, there exists only one pair of strongly connected 2-semi-chains.

\subsection{Binary Case}

In this subsection, we focus on the binary case, where $\Theta = \{0,1\}^K$. For convenience, we denote $\mathbf{G}_{\mathrm{diff}}^{(K,2)}$ as the differential graph over $\{0,1\}^K$. 

In this setting, we can partition $\Theta$ according to the value of the last entry:
$$\Theta_{\theta^{(K)} = 0}:=\{\theta \in \Theta : \theta^{(K)} = 0\}, \quad
\Theta_{\theta^{(K)} = 1} := \{\theta \in \Theta : \theta^{(K)} = 1\}.$$
Then, we can define $\mathbf{G}_{\mathrm{diff},a}^{(K,2)}$ as the subgraph of $\mathbf{G}_{\mathrm{diff}}^{(K,2)}$ restricted to $\Theta_{\theta^{(K)} = a}$, for all $a \in \{0,1\}$. Note that $\mathbf{G}_{\mathrm{diff},a}^{(K,2)}$ has the same structure as $\mathbf{G}_{\mathrm{diff}}^{(K-1,2)}$, for all $a \in \{1,2\}$. Moreover, in $\mathbf{G}_{\mathrm{diff}}^{(K,2)}$, aside from the edges in $\mathbf{G}_{\mathrm{diff},0}^{(K,2)} \cup \mathbf{G}_{\mathrm{diff},1}^{(K,2)}$, the remaining edges are between $\mathbf{G}_{\mathrm{diff},0}^{(K,2)}$ and $\mathbf{G}_{\mathrm{diff},1}^{(K,2)}$, connecting $(\theta^{(-K)},0)$ and $(\theta^{(-K)},1)$.
Hence, $\mathbf{G}_{\mathrm{diff}}^{(K,2)}$ can be constructed recursively by combining two copies of $\mathbf{G}_{\mathrm{diff}}^{(K-1,2)}$ and adding edges between nodes that differ only in their last coordinate. Figure~\ref{fig:strcutrue_dp_binary} illustrates this recursive structure.

\begin{figure}[t]
    \centering
    \includegraphics[width=\textwidth]{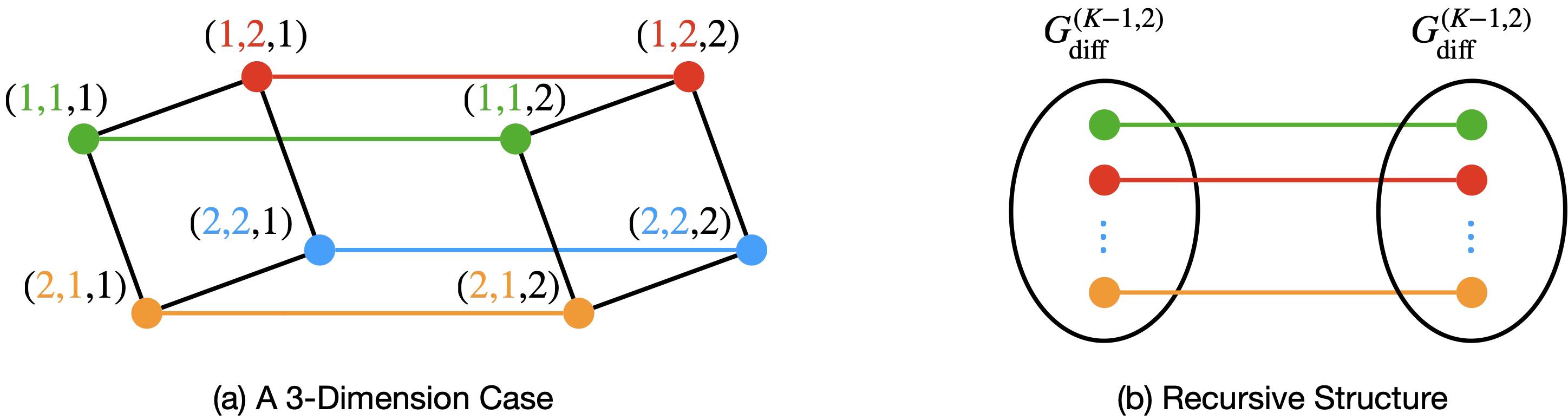}
    \caption{Recursive Structure of the Differential Graph in the Binary Case}
    \label{fig:strcutrue_dp_binary}
    
\end{figure}

Let $C = (\Theta_l)_{l=1}^L$ be an $L$-semi-chain.  The reverse semi-chain of $C$ is the ordered partition $C^{\mathrm{rev}} := (\Theta_{L-l+1})_{l=1}^L$, obtained by inverting the order of levels. Two semi-chains $C$ and $C'$ are said to be reverses of each other if $C' = C^{\mathrm{rev}}$.

\begin{proposition}\label{prp:dp_binary_unique}
    For any $K$-dimension differential graph in binary case, $\mathbf{G}_{\mathrm{diff}}^{(K,2)}$, there exists a unique strongly connected 2-semi-chain up to reversal.
\end{proposition}

When $K=2$, the differential graph $\mathbf{G}_{\mathrm{diff}}^{(2,2)}$ forms a cycle of four nodes. This case is illustrated in Figure~\ref{fig:upward_folding}, where we can easily observe that there exists a unique $2$-semi-chain up to reversal. The key feature of this pair of $2$-semi-chains is that all edges are between levels. Lemma~\ref{lemma:unique_binary} shows that this property guarantees the uniqueness of the $2$-semi-chain up to reversal. Moreover, due to the recursive structure of $\mathbf{G}_{\mathrm{diff}}^{(K,2)}$, this property is inherited from the $K$-dimensional case to the $(K+1)$-dimensional case.

\begin{lemma}\label{lemma:unique_binary}
    For any graph on $\Theta$, suppose there exists a strongly connected $2$-semi-chain with no within-level edges. Then this $2$-semi-chain is unique up to reversal.
\end{lemma}

\begin{proof}
    Let $(A,B)$ be the strongly connected $2$-semi-chain with no within-level edges, shown in Figure~\ref{fig:proof_dp_unique}(a). Take any other $2$-semi-chain $(C, D)$, shown in Figure~\ref{fig:proof_dp_unique}(b). Denote $$A_1 := C \cap A, \quad  B_1:=C\cap B, \quad A_2 := D \cap A, \quad B_2 := D \cap B.$$
    Since $(C,D)$ is different from $(A,B)$, $A_1 \cup B_2\neq \emptyset$, $A_2 \cup B_1 \neq \emptyset$.

    Since $(A,B)$ has no within-level edges, there are no edges between $A_1$ and $A_2$, nor between $B_1$ and $B_2$. Therefore, in the original graph, the edges from $A_1 \cup B_2$ to $A_2 \cup B_1$ are either be from $A_1$ to $B_1$ or from $B_2$ to $A_2$. But these edges are within-level with respect to $(C,D)$. Hence, $(C,D)$ cannot be strongly connected. Thus, the only strongly connected 2-semi-chains are $(A,B)$ and its reverse $(B,A)$, i.e. the 2-semi-chain is unique up to reversal.
\end{proof}

\begin{figure}[t]
    \centering
    \includegraphics[width=0.8\linewidth]{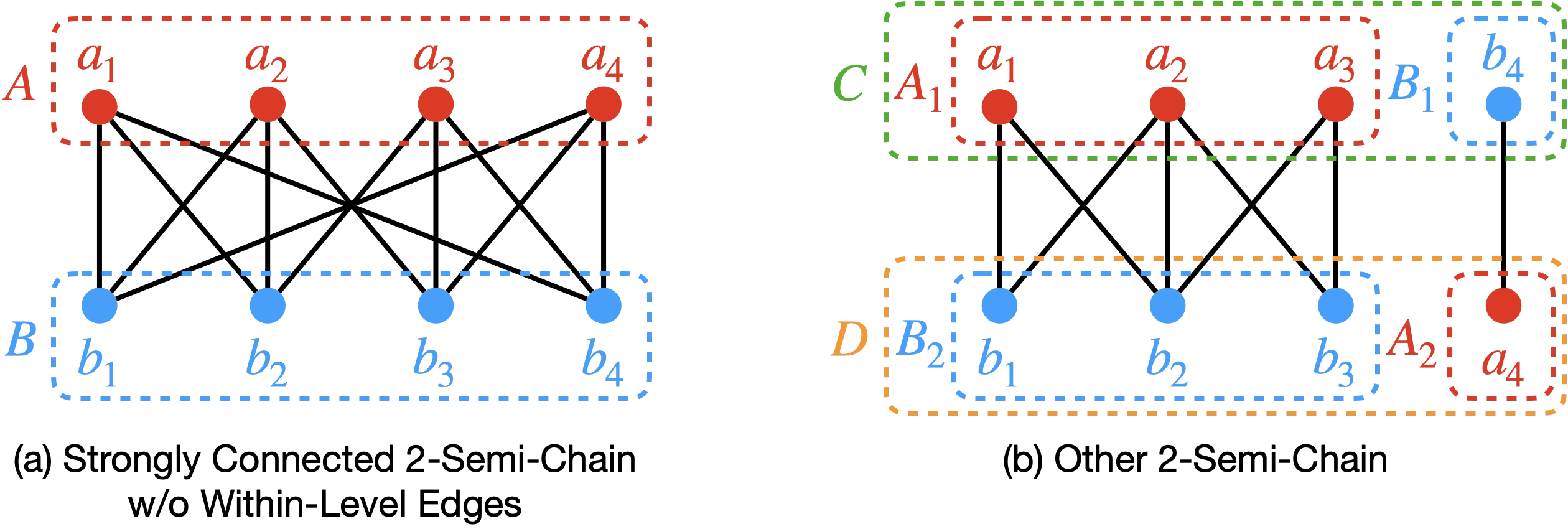}
    \caption{Illustration of Uniqueness in Lemma~\ref{lemma:unique_binary}}
    \label{fig:proof_dp_unique}
\end{figure}

The construction of this pair of $2$-semi-chains is straightforward. Starting from any node assigned to one level, all its neighboring nodes are placed in the opposite level, and this process continues recursively. Proposition~\ref{prp:dp_binary_unique} then implies that, in the binary differential privacy case, all extreme posteriors can be generated solely through unfolding.

\subsection{Two-Dimension Case}

In the two-dimensional case, only 2-semi-chains exist, and they can be generated from spanning trees of the graph. However, different spanning trees may yield the same 2-semi-chain. In this subsection, we further analyze the structural properties of 2-semi-chains in the two-dimensional case and propose a more efficient algorithm to generate all strongly connected 2-semi-chains without duplication.

Consider a two-dimensional dataset $\Theta = X \times Y$, where $X = \{x_1, \ldots, x_{n_1}\}$ and $Y = \{y_1, \ldots, y_{n_2}\}$ for some $n_1, n_2 \ge 2$. For each $\hat{x} \in X$ and $\hat{y} \in Y$, define the equivalence classes For each $\hat{x} \in X$ and $\hat{y} \in Y$, define $$[\hat{x}]: = \{(\hat{x}, y) : y \in Y\}, \quad [\hat{y}] := \{(x,\hat{y}) : x\in X\},$$
as the equivalent classes. By the definition of differential privacy, for any pair $(\theta_i, \theta_j) \in \Theta^2$, $(\theta_i,\theta_j) \in E$ if and only if there exists $\hat{x} \in X$ or $\hat{y} \in Y$ such that $\theta_i,\theta_j \in [\hat{x}]$ or $[\hat{y}]$.

Let $\mathbf{G}_{\mathrm{diff}}^{2}$ be the differential graph in two dimensions. If we partition the nodes according to $[x]$-equivalent classes, then:  
(i) for each $\hat{x} \in X$, the subgraph restricted to $[\hat{x}]$ is a complete graph;  
(ii) for each $\hat{y} \in Y$, the class $[\hat{y}]$ intersects every $[x]$-equivalent class at exactly one point and these points are connected with each other. Figure~\ref{fig:dp_2dim} provides illustrative examples.

\begin{figure}
    \centering
    \includegraphics[width = 0.8\textwidth]{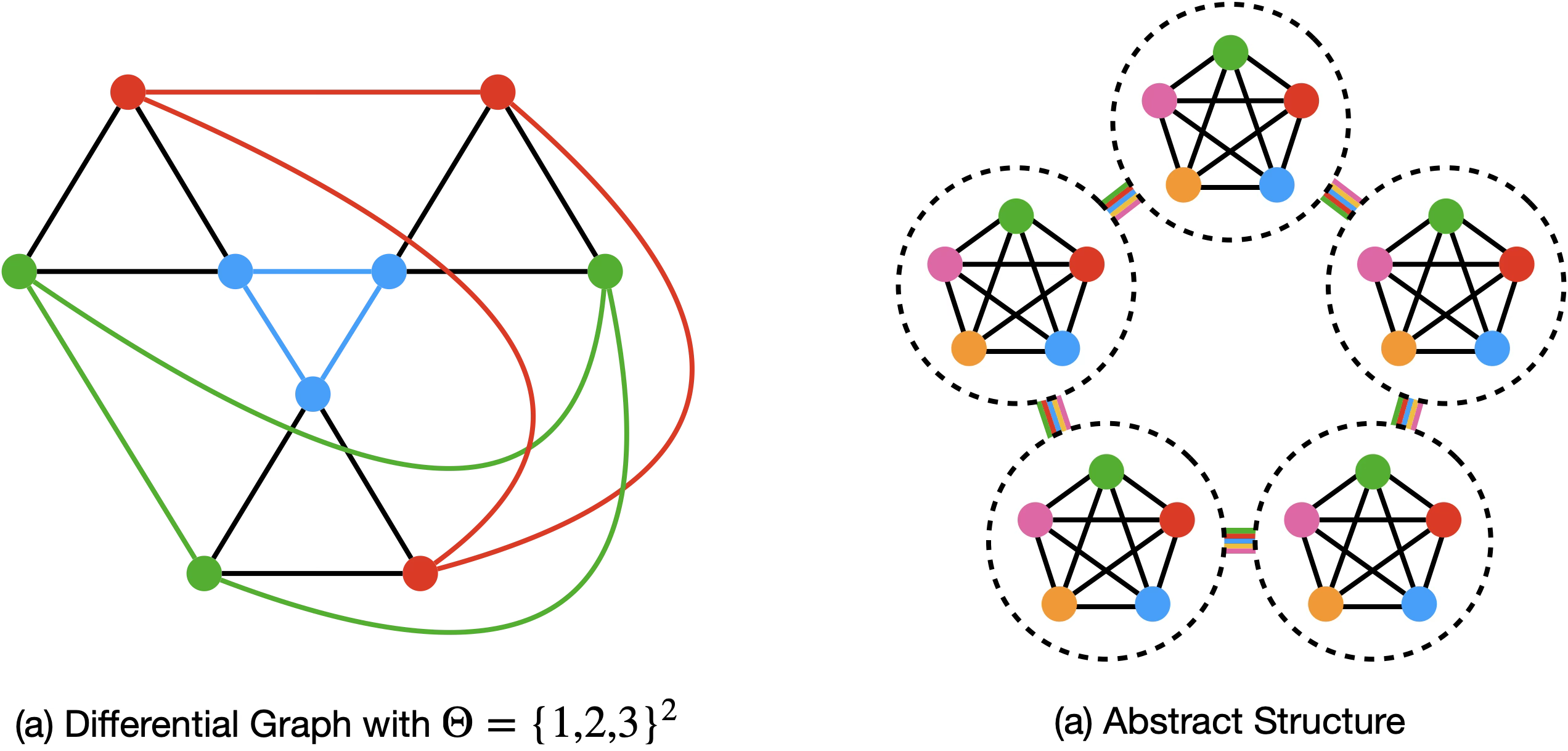}
    \caption{Structure of the  Differential Graph in the Two Dimension}
    \label{fig:dp_2dim}
\end{figure}

Given a 2-semi-chain of $\mathbf{G}_{\mathrm{diff}}^2$, it can be partitioned into $n_1$ components according to the $[x]$-equivalence classes. Each $[x]$-equivalence class corresponds to a \textit{division pattern} of $Y$, which divides the points in that class into two levels based on their $y$-coordinate values.

Let $\mathcal{D}$ denote the collection of division patterns of $Y$, which includes all ordered partitions of $Y$, as well as $(\emptyset, Y)$ and $(Y, \emptyset)$, representing the cases where all points in $Y$ are assigned to a single level.
For an equivalence class $[\hat{x}]$, we say that $[\hat{x}]$ in a 2-semi-chain follows the division pattern $\hat{d} = (Y_1, Y_2) \in \mathcal{D}$ if
$$[\hat{x}]_l = \{ (\hat{x}, y) : y \in Y_l \}, \quad \text{for all } l \in \{1,2\},$$
where $[\hat{x}]_l$ denotes the subset of $[\hat{x}]$ assigned to $l$-level in the 2-semi-chain.

Hence, to construct a 2-semi-chain: 
(i) select $n_1$ division patterns, and
(ii) assign these division patterns to $[x]$-equivalence classes. For a division pattern $d = (Y_1, Y_2)$, if $Y_1 = \emptyset$ or $Y_2 = \emptyset$, we say $d$ is undivided; otherwise, $d$ is divided. To ensure that the 2-semi-chain is strongly connected, there exists at least one divided division pattern. Then, we can distinguish the following three categories of strongly connected 2-semi-chains (Figure~\ref{fig:3_categories}): (Category I) no $[x]$-equivalent class is undivided; (Category II) exactly one level contains an undivided $[x]$-equivalent class; (Category III) both levels contain undivided $[x]$-equivalent classes.

\begin{figure}
    \centering
    \includegraphics[width=\textwidth]{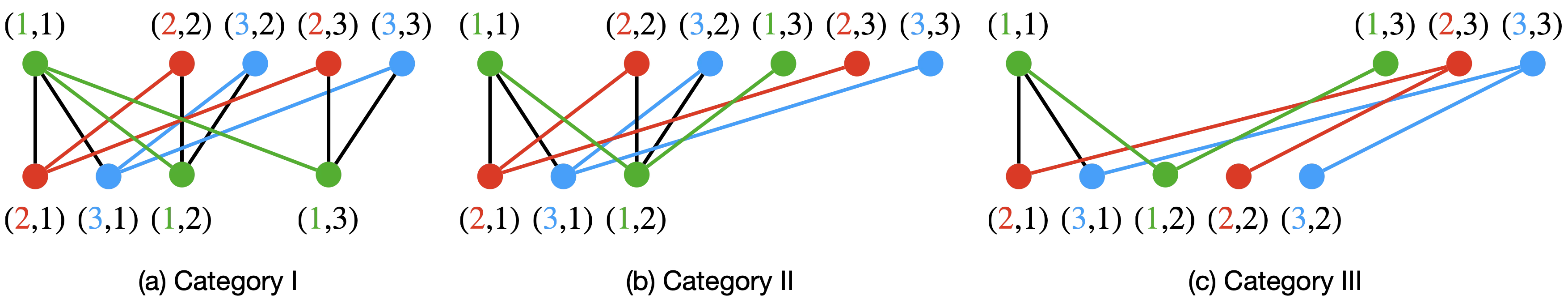}
    \caption{Illustration of the Three Categories of 2-Semi-Chains when $\Theta = \{1,2,3\}^2$}
    \label{fig:3_categories}
\end{figure}

\emph{Category I.}  Since the subgraph on each $[x]$-equivalent class is complete, when an $[x]$-equivalent class is divided, all inner nodes are connected with each other through between-level edges. In category I, as long as all equivalent classes are not partitioned in the same way, the 2-semi-chain are strongly connected: If there are two equivalent classes $[\hat{x}_1]$ and $[\hat{x}_2]$ that follow different patterns, then from $[\hat{x}_1]$ to $[\hat{x}_2]$ there exists at least one pair of nodes with the same $y$ connected by between-level edges. For the remaining $[x]$-equivalent classes, the division pattern must differ from both $[\hat{x}_1]$ and $[\hat{x}_2]$. Hence, the resulting 2-semi-chain is strongly connected.

\emph{Category II.} Since only one level contains an undivided $[x]$-equivalent class, the nodes in this class are connected only via within-level edges, as illustrated in Figure~\ref{fig:3_categories}(b). Therefore, to ensure that these nodes are connected by between-level edges, the opposite level must contain, for each $y \in Y$, at least one node in $[y]$.

\emph{Category III.} There is one divided equivalent class $[\hat{x}_1]$, and the first and second levels contain undivided equivalent classes $[\hat{x}_2]$ and $[\hat{x}_3]$, respectively. As illustrated in Figure~\ref{fig:3_categories}(c), each node in $[\hat{x}_1] \cup [\hat{x}_2] \cup [\hat{x}_3]$ is connected to the others through between-level edges. More importantly, for all $y \in Y$, each level contains at least one node in $[y]$. Hence, for the remaining $[x]$-equivalent classes, any division pattern can be chosen to maintain a strongly connected 2-semi-chain.

For a sequence of division patterns, $d = (d_{t_1}, \ldots, d_{t_{n_1}})$, if it can generate a strongly connected 2-semi-chain, we say $d$ is strongly connected division sequence. To construct an algorithm to generate all strongly connected division sequence, at first, we need code a number of different division pattern, $$\mathcal{D} = \{d_1, \ldots, d_{2^{n_2} -2}, d_{2^{n_2} -1}, d_{2^{n_2}}\},$$ where $d_{2^{n_2} - 1} = (Y, \emptyset)$, $d_{2^{n_2}} = (\emptyset, Y)$ and others corresponds to different ordered partition of $Y$. We say that $d = (d_{t_1}, \ldots, d_{t_{n_1}})$ is increasing if $t_1 \leq t_2 \leq \ldots \leq t_{n_1}$. We need only to generate all increasing, strongly connected division sequences, then applying these division patterns to $[x]$-equivalent classes according to arbitrary permutation, we can generate all strongly connected 2-semi-chains.

\begin{proposition}\label{prp:dp_2dim}
    Algorithm~\ref{algorithm:2_level_structure} generates all increasing, strongly connected sequences of division patterns for two-dimensional differential privacy, ensuring no duplication.
\end{proposition}

\begin{algorithm}
\caption{Construction of Increasing, Strongly Connected Division Sequences}
\label{algorithm:2_level_structure}
\KwIn{Integers $n_1$, $n_2$; set $Y$; and the collection of division patterns 
$\mathcal{D}$.}
\KwOut{Set $\mathcal{L}$ of all increasing, strongly connected division sequences}

Initialize $t_0 \gets 1$, $\mathcal{L} \gets \emptyset$\;

\SetKwFunction{Enumerate}{Enumerate}
\SetKwProg{Fn}{Function}{:}{}

\Fn{\Enumerate{$z, t_{z-1}, (t_1, \ldots, t_{z-1})$}}{
    \For{$t_z \in [t_{z-1}, 2^{n_2}-2]$}{
        \emph{[Category I]}
        \If{$z = n_1$ \KwAnd $\bigl|\{t_{1}, \ldots, t_{n_1}\}\bigr| \ge 2$}{
            Add $(d_{t_1}, \ldots, d_{t_{n_1}})$ to $\mathcal{L}$\;
        }
        \emph{[Category II]}
        \For{$l \in \{1,2\}$}{
            \If{$Y \subseteq \bigcup_{z'=1}^{z} d_{t_{z'}}[l]$ (for $d=(Y_1,Y_2)$, $d[l]=Y_l$)}{
                Add $(d_{t_1}, \ldots, d_{t_z}, 
                \underbrace{d_{2^{n_2}+l-2}, \ldots, d_{2^{n_2}+l-2}}_{n_1 - z})$ 
                to $\mathcal{L}$ \;
            }
        }
        \emph{[Category III]}
        \If{$z \leq n_1-2$}{
            \For{$n' \in [1, n_1 - z - 1]$}{
                Add $(d_{t_1}, \ldots, d_{t_z}, 
                \underbrace{d_{2^{n_2}-1}, \ldots, d_{2^{n_2}-1}}_{n'}, 
                \underbrace{d_{2^{n_2}}, \ldots, d_{2^{n_2}}}_{n_1 - n' - z})$ 
                to $\mathcal{L}$\;
            }
        }
        \If{$z < n_1$}{
            \Enumerate{$z+1, t_z, (t_1, \ldots, t_z)$}\;
        }
    }
}

\BlankLine
\Enumerate{$1, t_0, ()$}\;
\end{algorithm}

The logic of Algorithm~\ref{algorithm:2_level_structure} is as follows. Begin by choosing $d_{t_1}$ from the set of division patterns, then append $d_{t_2}$ with $t_2 \geq t_1$, and continue inductively with nondecreasing indices. Each time a new pattern is added, check whether the remaining positions can be filled by undivided patterns to produce a valid sequence in Category II or Category III; if so, output the corresponding sequence.
Figure~\ref{fig:placeholder} gives an example of running the algorithm.

\begin{figure}
    \centering
    \includegraphics[width=\linewidth]{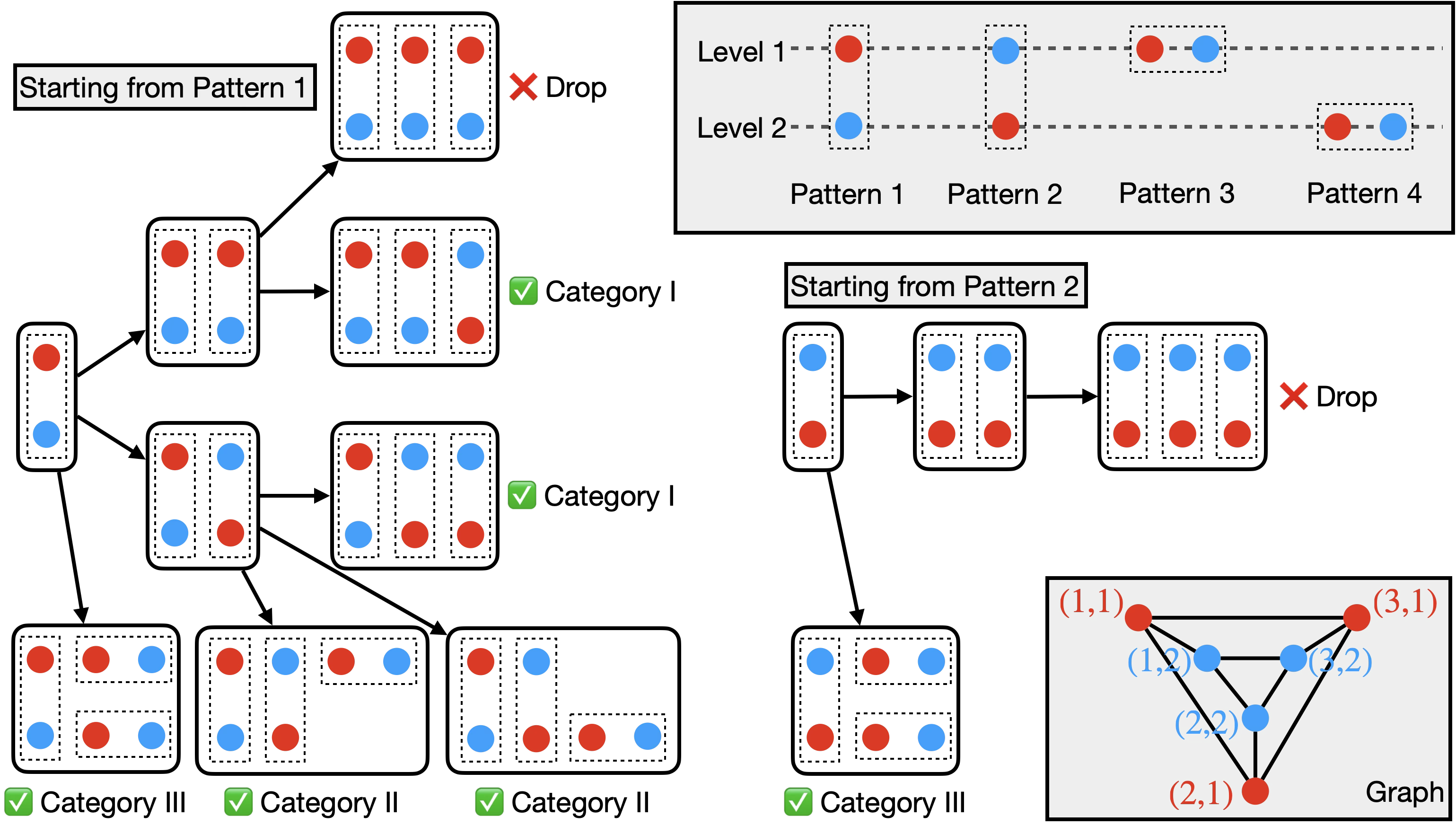}
    \caption{Illustration of Algorithm for $X = \{1,2,3\}$ and $Y = \{1,2\}$}
    \label{fig:placeholder}
\end{figure}

\section{Discussion and Future Work}
This paper characterizes the set of all feasible distributions subject to graph-based inferential privacy by establishing a connection between extreme posteriors and strongly connected semi-chains. We show that every extreme posterior can be generated by enumerating all strongly connected 2-semi-chains and subsequently applying the unfolding operation. For differential privacy, when each dimension of the dataset is binary, there exists a unique strongly connected 2-semi-chain up to reversal. In the two-dimensional case, only 2-semi-chains exist, and we propose an efficient algorithm to generate all strongly connected 2-semi-chains without duplication.

When the underlying data structure becomes more complex, the number of extreme posteriors, and hence the number of Blackwell-undominated privacy-preserving signals, grows rapidly. Determining which of these signals attains the optimum in a specific decision problem remains an open question. Future research will explore applications of the theoretical framework developed in this paper to such problems, including optimal information design under privacy constraints.

In addition, the microfoundations of privacy concepts warrant further investigation. Understanding how individuals or institutions perceive and trade off privacy in strategic settings could guide the selection of appropriate graph structures. More broadly, extending the characterization of the Blackwell frontier to other privacy notions beyond graph-based or differential privacy remains an important avenue for future work.

\newpage

\bibliographystyle{chicagoa}
\bibliography{refs}

\appendix

\section{Proofs}
\subsection*{Proof of Lemma \ref{lemma:extreme_points}}

Let $\mathcal{M}$ be a compact convex subset of $\Delta(\Theta)$, where $\Theta$ with its $\sigma$-algebra is a standard Borel space. 
To ensure the analysis is non-trivial, we assume that $\mu_0 \in \mathcal{M} \setminus \operatorname{ext}\mathcal{M}$.

In this subsection, we prove Lemma~\ref{lemma:extreme_points} for this more general $\mathcal{M}$.

\begin{proof}[Proof of Lemma \ref{lemma:extreme_points}]
    We first show (ii): A signal $\pi \in \overline{\Pi}_{\mathcal{M}}$ if and only if $\left<\pi\right>(\operatorname{ext}\mathcal{M}) = 1$. 
    ``If''. Suppose there exists another $\pi' \in \Pi_{\mathcal{M}}$ such that $\pi \prec \pi'$, then there exists a nondegenerate dilation $K: \Delta(\Theta) \to \Delta(\Delta(\Theta))$ such that for almost every $\mu \in \operatorname{supp}(\left<\pi\right>)$, $$\mu = \int_{\Delta(\Theta)} \nu K(d\nu |\mu).$$
    This means that for almost every $\mu \in \operatorname{supp}(\left<\pi\right>)$, it can be expressed by linear combination of $\nu \in \operatorname{supp}(\left<\pi'\right>)$.
    Since $K$ is nondegenerate, there exists a positive measure subset $A \subseteq \operatorname{supp}(\left<\pi\right>)$ such that $K(\mu|\mu) < 1$ for any $\mu \in A$. Hence, there is a positive measure subset $A$ such that $\mu \notin \operatorname{ext}\mathcal{M}$ for any $\mu \in E$, which is contradicted with $\left<\pi\right>(\operatorname{ext}\mathcal{M}) = 1$.

    ``Only if''. We construct a dilation from $\mathcal{M}$ to $\Delta(\operatorname{ext} \mathcal{M})$. Since $\Theta \subseteq \mathbb{R}^n$, $\Delta(\Theta)$ embeds into a locally convex space and endowed with the topology of weak convergence is metrizable. $\mathcal{M}$ is a compact convex subset of $\Delta(\Theta)$, which  is also metrizable. By Choquet’s Theorem (Theorem 10.7, p.168, \cite{simon2011convexity}), for any $\mu \in \mathcal{M}$, the set
    $$\Gamma(\mu) := \Bigl\{ P_\mu \in \Delta(\operatorname{ext}\mathcal{M}) : \mu = \int \nu  dP_\mu(\nu) \Bigr\}$$
    is nonempty. Moreover, $\Gamma(\mu)$ is closed in the weak-* topology. Define the barycenter map $B : \Delta(\mathcal{M}) \to \mathcal{M}$ by $B(P_\mu) = \int \nu  dP_\mu = \mu$. By \cite{simon2011convexity}, Theorem 9.1 (p.136), the map $B$ is continuous. Consequently, for any open set $U \subseteq \Delta(\operatorname{ext}\mathcal{M})$, $\Gamma^{-1}(U) = \{\mu \in \mathcal{M} : B^{-1}({\mu}) \cap U \neq \emptyset\} = B(U)$,
    which is an open set. Therefore, by the Kuratowski–Ryll-Nardzewski measurable selection theorem (Theorem 6.9.3, p.36, Vol. II, \cite{bogachev2007measure}), there exists a measurable selection $P_\mu^* : \mathcal{M} \to \Delta(\operatorname{ext}\mathcal{M})$ such that $\mu = \int \nu  dP_\mu^*(\nu)$ Hence, the map $D : \mu \mapsto P_\mu^*$ defines a dilation. If $\left<\pi\right>(\operatorname{ext}\mathcal{M}) \neq 1$, then $D$ is nondegenerate, which implies $\pi \notin \overline{\Pi}_{\mathcal{M}}$.

    For (i): A signal is ex post privacy-preserving if and only if it is Blackwell-dominated by some signal in $\overline{\Pi}_{\mathcal{M}}$. The ``if'' direction is immediate from the convexity of $\mathcal{M}$. The ``only if'' direction follows by the same argument used for the ``only if'' part of (ii).
\end{proof}

\subsection*{Proof of Theorem~\ref{thm:semi-chain}}

\begin{proof}[Proof of Lemma~\ref{lemma:feasible_w}]
    ``Only if''. By the definition, $\mu \in \mathcal{M}_\mathbf{G}^\varepsilon$ if
    \begin{equation}\label{eq:definition_M}
        e^{-\varepsilon}\frac{\mu(\theta_j)}{\mu_0(\theta_j)}\leq\frac{\mu(\theta_{i})}{\mu_0(\theta_{i})} \leq e^\varepsilon\frac{\mu(\theta_{j})}{\mu_0(\theta_{j})}, \quad \forall (\theta_i,\theta_j) \in \mathbf{G}, 
    \end{equation}
    which can be expressed as
    \begin{equation*}
        \frac{\mu(\theta_j)}{\mu_0(\theta_j)} = e^{w_{ij}\varepsilon} \frac{\mu(\theta_i)}{\mu_0(\theta_i)}, \text{ and } w_{ij} \in [-1,1], \quad \forall (\theta_i,\theta_j)\in \mathbf{G}.
    \end{equation*}
    Since $\mathbf{G}$ is connected, for any $(\theta_{j_1}, \theta_{j_r}) \in \Theta^2$, there is a path in $\mathbf{G}$, $\theta_{j_1}, \theta_{j_2}, \ldots, \theta_{j_r}$, such that
    \begin{equation*}
        \frac{\mu(\theta_{j_1})}{\mu_0(\theta_{j_1})} = e^{ \sum_{r'=1}^{r-1} w_{j_{r'}j_{r'+1}}\varepsilon} \frac{\mu(\theta_{j_r})}{\mu_0(\theta_{j_r})}.
    \end{equation*}
    If there is a $\theta_i \in \Theta$ such that $\mu(\theta_i) = 0$ then $\mu(\theta_j) = 0$ for all $\theta_j \in \Theta$ which contradicts that $\mu \in \Delta(\Theta)$. Therefore, $\delta_j = \ln \mu(\theta_j) - \ln \mu_0(\theta_j)$ is well-defined for all $\theta_j \in \Theta$. $\mu \in \mathcal{M}_{\mathbf{G}}^\varepsilon$ implies \eqref{eq:constraints} is valid and \eqref{eq:con_1} holds. Moreover, it is obvious  that \eqref{eq:con_2} $w_{ij}\varepsilon = \delta_i - \delta_j = -w_{ji} \varepsilon$ and \eqref{eq:con_3} $w_{j_1 j_r} \varepsilon= \delta_{j_r} - \delta_{j_1} = \sum_{r'=1}^{r-1}(\delta_{j_{r'+1}} - \delta_{j_{r'}}) = \sum_{r'=1}^{r-1}w_{j_{r'}j_{r'+1}} \varepsilon$. 
    
    ``If''. Suppose $\mathbf{W}$ satisfies these conditions. Consider a sequence $w_{j_1 j_2}, w_{j_2 j_3}, \ldots, w_{j_{J-1} j_J}$. The following $J$ independent linear constraints determine a unique $\hat{\mu}$:
    \begin{align}
        &\frac{\mu(\theta_{j_{r+1}})}{\mu_0(\theta_{j_{r+1}})}  - e^{w_{j_rj_{r+1}}\varepsilon}\frac{\mu(\theta_{j_r})}{\mu_0(\theta_{j_r})} = 0, \quad \forall r \in \{1,\ldots, J-1\}, \label{eq:constraint_2_prime}\\
        &\sum_{j=1}^N \mu(\theta_j) = 1, \label{eq:constraint_sum}
    \end{align}
    If there exists $\theta_j \in \Theta$ such that $\hat{\mu}(\theta_j) \leq 0$, then \eqref{eq:constraint_2_prime} will induces that $\hat{\mu}(\theta) \leq 0$ for all $\theta \in \Theta$, which contradicts with \eqref{eq:constraint_sum}. Thus, \eqref{eq:constraint_2_prime} can be rewritten as \eqref{eq:constraints} and  $\hat{\mu} \in \Delta(\Theta)$. Denote $\hat{\mathbf{W}}$ be the matrix induce by $\hat{\mu}$ which satisfies \eqref{eq:con_2}, \eqref{eq:con_3} and $\hat{w}_{j_rj_{r+1}} = w_{j_r j_{r+1}}$ for all $r \in \{1, \ldots, J-1\}$. Given the values of $w_{j_1 j_2}, w_{j_2 j_3}, \ldots, w_{j_{J-1} j_J}$, all other entries of $\mathbf{W}$ are endogenously and uniquely determined by \eqref{eq:con_2} and \eqref{eq:con_3}.
    Therefore, $\hat{\mathbf{W}} = \mathbf{W}$, then $\mathbf{W}$ can be obtained by $\hat{\mu}$. Moreover, since $\hat{\mathbf{W}} = \mathbf{W}$ satisfies \eqref{eq:con_1}, $\hat{\mu} \in \mathcal{M}_\mathbf{G}^\varepsilon$. 
\end{proof}

\begin{proof}[Proof of Lemma \ref{lemma:connecteed_acyclic}]
    By the definition, $\mu \in \mathcal{M}_\mathbf{G}^\varepsilon$ if and only if \eqref{eq:definition_M} (which induces $\mu(\theta_j) > 0$ for all $\theta_j \in \Theta$) and \eqref{eq:constraint_sum}.
    Since $\mu \in \mathbb{R}^J$,  $\mu \in \mathcal{M}_\mathbf{G}^\varepsilon $ is an extreme point of $\mathcal{M}_\mathbf{G}^\varepsilon$ if and only if a total of $J$ constraints are binding in \eqref{eq:definition_M} and \eqref{eq:constraint_sum} (Theorem 3.1, pp.102, \cite{bazaraa2011linear}). Because \eqref{eq:constraint_sum} is always binding, it remains to identify $J-1$ binding constraints from \eqref{eq:definition_M}. For each pair of constraints in \eqref{eq:definition_M}, $e^{-\varepsilon}\frac{\mu(\theta_j)}{\mu_0(\theta_j)}\leq\frac{\mu(\theta_{i})}{\mu_0(\theta_{i})} \leq e^\varepsilon\frac{\mu(\theta_{j})}{\mu_0(\theta_{j})}$, it corresponds to $w_{ij} = -w_{ji} \in [-1,1]$ and the two inequalities cannot be binding simultaneously. Hence, finding $J-1$ independent binding constraints in \eqref{eq:definition_M} is equivalent to identifying $J-1$ independent $w_{ij} = \pm 1$ in $\mathbf{G}$.

    Now, we show that identifying $J-1$ independent $w_{ij} = \pm 1$ for some $(\theta_i, \theta_j) \in \mathbf{G}$ is equivalent to find a spanning tree $\mathbf{T}$ of $\mathbf{G}$ such that $w_{ij} =\pm 1$ for every $(\theta_{i}, \theta_{j}) \in \mathbf{T}$. 
    
    ``If''. For a tree $\mathbf{T}$, 
    by Lemma~\ref{lemma:feasible_w}, the collection of $\mathcal{W}_\mathbf{T}:=\{w_{ij}:(\theta_i,\theta_j) \in \mathbf{T} \text{ and } i < j\}$ is not independent if and only there exists some $\hat{w} \in \mathcal{W}_\mathbf{T}$ that can be expressed as a linear combination of the others through the equations in \eqref{eq:con_2} and \eqref{eq:con_3}. By construction of $\mathcal{W}_\mathbf{T}$, the constraints in \eqref{eq:con_2} do not apply. Moreover, since $\mathbf{T}$ is acyclic, \eqref{eq:con_3} does not apply either. Thus, $\mathcal{W}_\mathbf{T}$ is independent. Because $\mathbf{T}$ has $J-1$ edges, we have $|\mathcal{W}_\mathbf{T}| = J-1$.  Hence, $\mathcal{W}_\mathbf{T}$ identifies $J-1$ independent constraints $w_{ij} = \pm 1$ for $(\theta_i, \theta_j) \in \mathbf{T}$.
    
    ``Only if''. Suppose no such spanning tree exists. Assume there are $J-1$ pairs of constraints $w_{ij} = \pm 1 = -w_{ji}$ for some $(\theta_i, \theta_j) \in \mathbf{G}$; otherwise, it is trivial. 
    Consider any subgraph with $J-1$ edges such that $w_{ij} = \pm 1$ for every edge. Such a subgraph necessarily contains a cycle.  By \eqref{eq:con_3}, within this cycle one weight can be written as a linear combination of the others. Hence, the cycle reduces the number of independent constraints, so the subgraph yields at most $J-2$ independent weights.
\end{proof}

\begin{lemma}\label{lemma:pm1_0}
    Suppose $\mathbf{W}$ corresponds to some $\mu\in \mathcal{M}_\mathbf{G}^\varepsilon$. If $\mu \in \operatorname{ext}\mathcal{M}_\mathbf{G}^\varepsilon$, then $w_{ij} \in \{\pm 1, 0\}$ for every $(\theta_{i}, \theta_{j}) \in \mathbf{G}$.
\end{lemma}

\begin{proof}[Proof of Lemma~\ref{lemma:pm1_0}]
    Suppose $\mu \in \operatorname{ext} \mathcal{M}_\mathbf{G}^\varepsilon$, then according to Lemma~\ref{lemma:connecteed_acyclic}, there exists a spanning tree $\mathbf{T}$ such that $w_{ij} = \pm 1$ for all $(\theta_i, \theta_j) \in \mathbf{T}$. For any other $(\theta_{i'}, \theta_{j'}) \in \mathbf{G} \setminus \mathbf{T}$, there exists a path from $\theta_{i'}$ to $\theta_{j'}$ in $\mathbf{T}$. Thus, $w_{i'j'}$ is determined by the sum of the weights along this path. Since for any $(\theta_i, \theta_j) \in \mathbf{T}$, $w_{ij} = \pm 1$, then $w_{i'j'} \in \{\pm 1, 0\}$.
\end{proof}

\begin{lemma}\label{lemma:c_C}
    For any $L$-semi-chain $C=(\Theta_l)_{l =1}^L$, $L \geq 2$, the matrix $\mathbf{W}$ satisfying
    \begin{equation*}
        w_{i j} =
        \begin{cases}
        1 & \text{if there exists } l \in \{1,\ldots, L-1\} \text{ such that } \theta_{i} \in \Theta_l, \theta_{j} \in \Theta_{l+1} \\
        0 & \text{if there exists } l \in\{1,\ldots,L\} \text{ such that } \theta_{i}, \theta_{j} \in \Theta_l
        \end{cases}.
    \end{equation*} can correspond to some $\mu \in \mathcal{M}^\varepsilon_\mathbf{G}$.
\end{lemma}

\begin{proof}[Proof of Lemma \ref{lemma:c_C}]

    The remaining entries of $\mathbf{W}$ are determined by \eqref{eq:con_2} and \eqref{eq:con_3}. Since constraints \eqref{eq:con_1} and \eqref{eq:con_2} are automatically satisfied, it remains to show that constraint \eqref{eq:con_3} is also satisfied. In particular, for any cycle in the graph $\mathbf{G}$, $\theta_{j_1}, \theta_{j_2}, \ldots, \theta_{j_r}, \theta_{j_{r+1} = \theta_{j_1}}$, $\sum_{r'=1}^{r} w_{j_r' j_{r'+1}} = 0.$

    Consider an arbitrary cycle. Its edges can be divided into two types: (i) Within-level edges. The sum of their weights is always zero by construction. (ii) Between-level edges. These edges induce a cycle in the partition graph, where partitions $\{\Theta_l\}_{l=1}^L$ are vertices and a pair of vertices are linked if and only if they are adjacent levels. Thus, the partition graph is a line from top to bottom. To form a cycle in this line graph, every downward step must eventually be matched by a upward step. Since the weight $w$ of an edge from $l$-th level $\Theta_l$ to $(l-1)$-th level $\Theta_{l+1}$ is the negative of the weight of the reverse edge (from $(l-1)$-th level to $l$-th level), the total weight of edges between levels cancels out to zero. Accordingly, in any cycle, the total weight is $0$.
\end{proof}

\begin{proof}[Proof of Theorem \ref{thm:semi-chain}]
    ``If''. Due to Lemma~\ref{lemma:c_C}, such $\mathbf{W}$ can correspond to some $\mu\in\mathcal{M}^\varepsilon_\mathbf{G}$ and therefore satisfy \eqref{eq:con_1}, \eqref{eq:con_2}, \eqref{eq:con_3}. The strong connectedness of $C$ implies that the subgraph post deleting all within-level edges is connected. Note that if $\theta_i$ and $\theta_j$ are linked in the subgraph, then $w_{ij} = \pm 1$. The finite connected subgraph always has at least one spanning tree $\mathbf{T}$ (Theorem 1.5.1, pp.14, \citealt{diestel2025graph}), such that any pair $(\theta_i, \theta_j) \in \mathbf{T}$, $w_{ij} = \pm 1$. Lemma \ref{lemma:connecteed_acyclic} implies that $\mu \in \operatorname{ext} \mathcal{M}_\mathbf{G}^\varepsilon$.
    

    ``Only if''. Suppose $\mu \in \operatorname{ext} \mathcal{M}_\mathbf{G}^\varepsilon$ and let $\mathbf{W}$ be the corresponding matrix, satisfying \eqref{eq:con_1}, \eqref{eq:con_2}, \eqref{eq:con_3}. By Lemma~\ref{lemma:connecteed_acyclic}, there exists a spanning tree $\mathbf{T}$ such that $w_{ij} = \pm 1$ whenever $(\theta_i, \theta_j) \in \mathbf{T}$. Begin with an arbitrary root node $\theta_i$ and place it in some level. For each neighbor $\theta_j \in N(\theta_i)$, if $w_{ij} = 1$ then assign $\theta_j$ to the level immediately above $\theta_1$; if $w_{ij} = -1$, assign $\theta_j$ to the level immediately below $\theta_i$. Proceed inductively: for any newly assigned node, allocate its neighbors according to the same rule. Since $T$ is a spanning tree, each node is assigned exactly once, and this procedure yields a well-defined ordered partition $(\Theta_l)_{l=1}^L$.  

    Now consider any $\theta_i \in \Theta_l$ and $\theta_j \in \Theta_{l'}$. There is a path $(\theta_{i_k})^n_0$ in $\mathbf{T}$, with $\theta_{i_0} = \theta_i$, $\theta_{i_n} = \theta_j$. Regarding the path, it is required that, if $\theta_{i_k}\in \Theta_l$, then $\theta_{i_{k+1}}\in \Theta_{l-1}\cup\Theta_{l+1}$, and $w_{i_{k}i_{k+1}} = 1 (-1)$ if $\theta_{i_{k+1}}\in \Theta_{l+1} (\theta_{i_{k+1}}\in \Theta_{l-1})$. Constraint \eqref{eq:con_3} then implies that 
    \[
    w_{ij} = w_{i_0 i_{n}} = \sum^{n-1}_{k=0}{w_{i_k i_{k+1}}} = l' - l
    \]
    If $l=l^\prime$, then $w_{ij} = 0$, and if $l^\prime - l=1$, then $w_{ij} = 1$. Therefore, $\mathbf{W}$ satisfies the conditions stated in the Theorem. Besides, if $|l-l'|\geq 2$, then $(\theta_i,\theta_j)\notin \mathbf{G}$, otherwise constraint \eqref{eq:con_1} will be violated. Therefore, the ordered partition $(\Theta_l)_{l=1}^L$ forms a semi-chain. Finally, since the spanning tree is contained in the subgraph post deleting links within same partitions, the semi-chain is  strongly connected.

\end{proof}

\subsection*{Proof of Theorem~\ref{thm:2-semi-chain}}

\begin{definition}
Let $C^L = (\Theta_l)_{l = 1}^L$ be an $L$-semi-chain for some $L \geq 3$.  
\begin{enumerate}
    \item The ordered partition $(\Theta_2, \Theta_1 \cup \Theta_3, (\Theta_l)_{i=4}^L)$ is called an downward folding of $C^L$.
    \item The ordered partition $((\Theta_l)_{l=1}^{L-3}, \Theta_{L-2} \cup \Theta_L, \Theta_{L-1})$  is called an upward folding of $C^L$.
    \item The $L_{-}$-semi-chain $C^{L_{-}}$, for some $2\leq L_{-} <L$, is obtained from $C^L$ by successive downward (resp. upward) folding if there exists a sequence $C^{L_{-}}, C^{L_{-}+1}, \ldots, C^L$ such that each $C^{L'}$ is an downward (resp. upward) folding of $C^{L'+1}$, for all $L' \in \{L_{-}, \ldots, L-1\}$.
\end{enumerate}
\end{definition}

It is clear that the downward (or upward) folding of an $L$-semi-chain with $L \geq 3$ is a semi-chain, and this folding is unique.

\begin{lemma}\label{lemma:succession_connected}
    Given an $L$-semi-chain $C^{L}$ for some $L \geq 3$, $C^{L-1}$ is obtained from $C^L$ by downward (or upward) folding. $C^L$ is strongly connected if and only if $C^{L-1}$ is strongly connected.
\end{lemma}  
\begin{proof}[Proof of Lemma~\ref{lemma:succession_connected}]
    W.l.o.g., we restrict attention to downward folding.  Let $C^L=(\Theta_l)_{l=1}^L$, then $C^{L-1} = (\Theta_2, \Theta_1\cup \Theta_3, (\Theta_l)_{l=4}^L)$. Suppose $C^L$ is strongly connected. After folding $\Theta_1$ into $\Theta_3$, the only modification is that the signs of the edges between $\Theta_1$ and $\Theta_2$ are reversed. Importantly, no between-level edge is converted into a within-level edge. Hence, every path that certifies strong connectivity in $C^L$ remains valid in $C^{L-1}$. Therefore, $C^{L-1}$ is strongly connected. Conversely, suppose $C^{L-1}$ is strongly connected. When unfolding by placing $\Theta_1$ back into the first level, by the definition of $C^L$, $\Theta_1$ is adjacent only to $\Theta_2$ and has no edges to $\Theta_3 \cup \Theta_4$. Thus, reinstating $\Theta_1$ does not create any within-level edges and preserves all existing between-level edges. Consequently, strong connectivity of $C^{L-1}$ implies strong connectivity of $C^L$.
\end{proof}

\begin{lemma}\label{lemma:split_fold}
    $C^{L+1}$ is an upward (resp. downward) unfolding of $C^L$ if and only if $C^L$ is the downward (resp. upward) folding of $C^{L+1}$.
\end{lemma}

\begin{proof}[Proof of Lemma~\ref{lemma:split_fold}]
    Let $C^{L+1} = (\Theta_l)^L_{l=1}$ and $C^L = (\hat{\Theta}_l)_{l=1}^L$. 
    Suppose $C^{L+1}$ is an upward unfolding of $C^L$. Then there exists a partition $\{\hat{\Theta}_2^1, \hat{\Theta}_2^2\}$ such that $\Theta_1 = \hat{\Theta}_2^1$ and $\Theta_3 = \hat{\Theta}_2^2$. Then, by merging $\Theta_1$ back into $\Theta_3$, we recover $C^L$. Conversely, suppose $C^L$ is obtained from $C^{L+1}$ by downward folding. Then
    $\hat{\Theta}_2 = \Theta_1 \cup \Theta_3$.
    By the definition of $C^{L+1}$, there is no edge from $\Theta_1$ to $\Theta_3 \cup (\hat{\Theta}_3 = \Theta_4)$. Hence, $C^{L+1}$ can be obtained from $C^L$ by dividing $\hat{\Theta}_2$ into $\Theta_1$ and $\Theta_3$, and then assigning $\Theta_1$ as the new first level.
\end{proof}

\begin{proof}[Proof of Theorem \ref{thm:2-semi-chain}]
    Through recursion, it is easy to show that an ordered partition $(\Theta_l)_{l=1}^L$ obtained from a $2$-semi-chain by successive upward unfolding is a semi-chain. For any $C^L$, there exists a sequence of decreasing-level semi-chain $C^{L-1}, C^{L-2}, \ldots, C^{2}$, where each $C^{L'}$ is obtained from $C^{L'+1}$ by downward folding. By Lemma~\ref{lemma:split_fold}, it follows that $C^L$ can be obtained by successive upward unfolding starting from $C^2$. Moreover, by Lemma~\ref{lemma:succession_connected}, $C^L$ is strongly connected if and only if the $C^2$ from which it is generated is strongly connected.
\end{proof}

\subsection*{Proof of Proposition~\ref{prop:constrain_level_dp}}

\begin{lemma}\label{lemma:K+1_dp}
    For the differential graph $\mathbf{G}_{\mathrm{diff}}$ in $K$-dimension case, there exists a strongly connected $(K+1)$-semi-chain.
\end{lemma}
\begin{proof}[Proof of Lemma~\ref{lemma:K+1_dp}]
    The proof is based on mathematical induction and constructive methods. Some notations are introduced beforehand. Within this proof, denote $\Theta^{[m]}:=\prod^m_{k=1}\Theta^{(k)}$ the set of $m$-dimensional states, with a typical element $\theta^{[m]} = (\theta^{(l)})^m_{l=1}$.  
    
    Case $K=1$. Clearly, choosing any single point in $\Theta$ as the first level and placing the remaining points in the second level yields a 2-semi-chain. Moreover, since when $K=1$, $\mathbf{G}_{\mathrm{diff}}$ is complete, this 2-semi-chain is strongly connected.

    Suppose for $K = m-1$, there exists a strongly connected $m$-semi-chain,
    denoted by $C^m = (\Theta^{[m-1]}_l)_{l=1}^m$. For $K = m$, fix $\hat{\theta}^{(m)} \in \Theta^{(m)}$ and construct $C^{m+1} = (\Theta^{[m]}_l)_{l=1}^{m+1}$ as follows: 
\begin{align*}
    \Theta_{1}^{[m]} & =\left\{ (\theta^{[m-1]},\hat{\theta}^{(m)}):\theta^{[m-1]}\in\Theta_{1}^{[m-1]}\right\} ,\\
    \Theta_{l}^{[m]}=\left\{ (\theta^{[m-1]},\hat{\theta}^{(m)}):\theta^{[m-1]}\in\Theta_{l}^{[m-1]}\right\}  & \bigcup\left\{ (\theta^{[m-1]},\theta^{(m)}):\theta^{[m-1]}\in\Theta_{l-1}^{[m-1]},\theta^{(m)}\ne\hat{\theta}^{(m)}\right\} ,\forall l\in\{2,\cdot\cdot\cdot,m\}\\
    \Theta_{m+1}^{[m]} & =\left\{ (\theta^{[m-1]},\theta^{(m)}):\theta^{[m-1]}\in\Theta_{m}^{[m-1]},\theta^{(m)}\ne\hat{\theta}^{(m)}\right\} 
\end{align*}
Fix some $\tilde{\theta}^{(m)} \in \Theta^{(m)}$, denote 
\[
C_{\tilde{\theta}^{(m)}}^{m+1}:=\left\{ \Theta_{l}^{[m]}\cap\{\theta^{[m]}:\theta^{(m)}=\tilde{\theta}^{(m)}\}\right\} 
\]
the ordered partition restricting to nodes with $m$-th coordinate $\tilde{\theta}^{(m)}$. By construction, we have that 
\[
\begin{cases}
(\theta^{[m-1]},\theta^{(m)})\in\Theta_{l+1}^{[m]}\ \iff\theta^{[m-1]}\in\Theta_{l}^{[m-1]} & \text{if }\theta^{(m)}\ne\hat{\theta}^{(m)}\\
(\theta^{[m-1]},\hat{\theta}^{(m)})\in\Theta_{l}^{[m]}\ \iff\theta^{[m-1]}\in\Theta_{l}^{[m-1]}
\end{cases}
\]
In this sense, $C_{\tilde{\theta}^{(m)}}^{m+1}$ (for the subgraph of $G$, formed by nodes with $m$-th coordinate $\tilde{\theta}^{(m)}$) is of same structure as $C^m$, for any $\tilde{\theta}^{(m)}\in \Theta^{(m)}$. Therefore, the edge between a pair with fixed $m$-th coordinate is either within same level or between adjacent levels, since $C^m$ is a semi-chain. On the other hand, the edge between a pair with different $m$-th coordinates (and therefore same first $m-1$ coordinate) is either within same level (if both $m$-th coordinates differs from $\hat{\theta}^{(m)}$) or between adjacent levels (if one of the $m$-th coordinates is $\hat{\theta}^{(m)}$). Accordingly, the ordered partition $C^{m+1}$ is an $(m+1)$-semi-chain. 

It remains to prove the strong connectivity of $C^{m+1}$. On one hand,  $C_{\tilde{\theta}^{(m)}}^{m+1}$ inherits the strong connectivity from $C^m$. On the other hand, there is an edge between $(\theta^{[m-1]},\theta^{(m)})$ and $(\theta^{[m-1]},\hat{\theta^{(m)}})$, which is between adjacent levels. Accordingly, $C^{m+1}$ is strongly connected.



    
\end{proof}

\begin{proof}[Proof of Proposition~\ref{prop:constrain_level_dp}]
    ``If''. By Lemma~\ref{lemma:K+1_dp}, a strongly connected $(K+1)$-semi-chain exists. By Lemma~\ref{lemma:succession_connected}, successive downward folding of this strongly connected $(K+1)$-semi-chain generates a strongly connected $L$-semi-chain for any $L \le K$.

    ``Only if''. Suppose there exists an $L$-semi-chain. Then the shortest path connecting a node in the first level to a node in the last level must have length at least $L-1$. However, for any two points in $\Theta$, they can differ in at most $K$ dimensions, so, by the definition of $\mathbf{G}_\mathrm{diff}$, the shortest path connecting them has length at most $K$. Therefore, the maximum possible length of a semi-chain is $K+1$.
\end{proof}

\subsection*{Proof of Proposition~\ref{prp:dp_binary_unique}}

\begin{lemma}\label{lemma:transitivity_binary}
    Suppose $\mathbf{G}_{\mathrm{diff}}^{(K-1,2)}$ has a strongly connected $2$-semi-chain $(A,B)$ in which there is no within-level edges. Then, $((A,0) \cup (B,1), (B,0) \cup (A,1))$ is a $2$-semi-chain with no within-level edges of $\mathbf{G}_{\mathrm{diff}}^{(K,2)}$, where $(A,0) := \{(\theta^{(-K)}, 0) : \theta^{(-K)} \in A\}$ and similarly for others.
\end{lemma}
\begin{proof}
    First, since $(A,B)$ is an ordered partition of $\{0,1\}^{n-1}$, $((A,0) \cup (B,1), (B,0) \cup (A,1))$ is an ordered partition of $\{0,1\}^n$. Moreover, any $2$-ordered partition is $2$-semi-chain.
    
    Strong connectivity: By the strong connectivity of $(A,B)$, any pair of nodes in $(A,0) \cup (B,0)$ is connected by a between-level path (i.e., a path along within-level edges), and the same holds for $(A,1) \cup (B,1)$. For any $(\theta_i^{(-K)},0) \in (A,0)$ and $(\theta_j^{(-K)},1) \in (A,1) \cup (B,1)$, there is a within-level path from $(\theta_i^{(-K)},0)$ to $(\theta_i^{(-K)},1)$, and since $(\theta_i^{(-K)},1) \in (A,1)$, it has a between-level path to $(\theta_j^{(-K)},1)$. Hence $(\theta_i^{(-K)},0)$ and $(\theta_j^{(-K)},1)$ are connected via a path consisting entirely of between-level edges.
    Thus, every node in $(A,0)$ is connected via between-level edges with any other node, and by the same reasoning the result holds for $(A,1)$, $(B,0)$ and $(B,1)$. Consequently, $((A,0) \cup (B,1), (B,0) \cup (A,1))$ is a strongly connected.

    No within-level edge: $\mathbf{G}_{\mathrm{diff}}^{(K,2)}$ can be divided into two $\mathbf{G}_{\mathrm{diff}}^{(K-1,2)}$. The edges of $\mathbf{G}_\mathrm{diff}^{(K,2)}$ within the same $\mathbf{G}_{\mathrm{diff}}^{(K-1,2)}$ are between-level and other edges connecting $\theta = (\theta^{(-K)},0)$ and $\theta' = (\theta^{(-K)},1)$ are also between-level, since $(A,0)$, $(B,1)$ are in the first level while $(A,1)$ and $(B,0)$ are in the second level.
\end{proof}

\begin{proof}[Proof of Proposition \ref{prp:dp_binary_unique}]
    When $K=1$, there are only two points in $\Theta = \{0, 1\}$. Hence, the unique (up to reversal) strongly connected 2-semi-chain is $(\{0\},\{1\})$ which satisfies the assumption in Lemma \ref{lemma:unique_binary}. Hence, since the transitivity in Lemma \ref{lemma:transitivity_binary} and uniqueness in Lemma \ref{lemma:unique_binary}, for any $K$, there is a unique strongly connected 2-semi-chain.
\end{proof}

\subsection*{Proof of Proposition~\ref{prp:dp_2dim}}

\begin{lemma}\label{lemma:divided}
    For any strongly connected 2-semi-chain of $\mathbf{G}_{\mathrm{diff}}^{2}$, there exists at least one  $\hat{x} \in X$ such that $[\hat{x}]$ is divided.
\end{lemma}
\begin{proof}[Proof of Lemma~\ref{lemma:divided}]
    Suppose for any $\hat{x} \in X$,  every $\theta \in [\hat{x}]$ is in the same level. Then, given $\hat{x}$, for any $y, y'\in Y$, edge $((\hat{x}, y), (\hat{x},y'))$ is within-level. Suppose there is a path between $(x_1, y_1)$ and $(x_2, y_2)$, then there must exist $x' \in X$ such that edge $((x',y_1),(x',y_2))$ is in this path. Therefore, there does not exist a path from $(x_1, y_1)$ to $(x_2, y_2)$ without a within-level edge. It contradicts with strong connectivity.
\end{proof}

\begin{lemma}\label{lemma:two_division}
    After allocating any number of divided $[x]$, the induced sub-2-semi-chain is strongly connected if and only if there is at least two $[x]$, the division pattern are different.
\end{lemma}
\begin{proof}
    ``Only if''. Suppose there is only one division pattern, then the edges in the same $[\hat{y}]$ are within-level. Hence $(x_1, \hat{y})$ and $(x_2, \hat{y})$ are not connected by a between-level path.

    ``If''. Assume $[x_1]$ and $[x_2]$ have different division methods. First, since $[x_1]$ and $[x_2]$ are divided and subgraph on the them are complete, the nodes within the same equivalent class are connected by a between-level path. Second, since the division methods are different, there is a $\hat{y} \in Y$ such that $(x_1, \hat{y})$ and $(x_2, \hat{y})$ are in the different level. Then, $(x_1, \hat{y})$ and $(x_2, \hat{y})$ are connected by a between-level path. Therefore, the sub-2-semi-chain restricted on $[x_1]$ and $[x_2]$ are strongly connected. Furthermore, for the sequential divided $[x_l]$, its division method must be different with one of $[x_1]$ and $[x_2]$. Since the same reason, the sub-2-semi-chain restricted on these divided $[x]$ are strongly connected.
\end{proof}

\begin{lemma}\label{lemma:feasible_same_level}
    A 2-semi-chain with $m$ divided $[x]$-equivalent classes and $n_1 - m$ undivided ones in the first (resp. second) level is strongly connected if and only if, after allocating the $m$ divided $[x]$,  for every $y \in Y$, there exists $\hat{\theta}_y$ in the second (resp. first) level such that $\hat{\theta}_y \in [y]$, and all these nodes $\{\hat{\theta}_y\}_{y\in Y}$ are connected by a between-level path.
\end{lemma}

\begin{proof}[Proof of Lemma~\ref{lemma:feasible_same_level}]
    ``If''. Suppose after allocating all divided $[x]$, in second level, for each $\hat{y} \in  Y$, there is $\theta_{\hat{y}} \in [\hat{y}]$. Then, there are at least two divided $[x]$ with different division pattern and according to Lemma~\ref{lemma:two_division}, these divided $[x]$ are strongly connected. Furthermore, since for each point in the rest undivided $[x]$ in the first level, it connects with the corresponding $\hat{\theta}_{y}$ in the second level, the 2-semi-chain is strongly connected.

    ``Only if''. Suppose in the second level, there is no point in $[\hat{y}]$. Then, for $[x_{m+1}]$ which is totally allocated in the first level, $(x_{m+1}, \hat{y})$ only have edges with nodes in $[x_{m+1}]\cup[\hat{y}]$ which are all within the first level. Hence, 2-semi-chain is not strongly connected.
\end{proof}

\begin{lemma}\label{lemma:2_level_undivision}
    Any 2-semi-chain contains one divided $[\hat{x}_1]$, one undivided $[\hat{x}_2]$ in the first level, one undivided $[\hat{x}_3]$ in the second level, then it is strongly connected.
\end{lemma}
\begin{proof}[Proof of Lemma~\ref{lemma:2_level_undivision}]
    Since $[\hat{x}_1]$ is divided and subgraph on $[\hat{x}_1]$ is complete, all nodes in $[\hat{x}_1]$ are connected by a between-level path. Let $([\hat{x}_1]_1, [\hat{x}_1]_2)$ be the partition of $[\hat{x}_1]$. Denote $Y_1 := \{y \in Y:  (\hat{x}_1,y) \in [\hat{x}_1]_1\}$ and $Y_2 := \{y \in Y:  (\hat{x}_1,y) \in [\hat{x}_1]_2\}$. 
    
    Since the whole $[\hat{x}_2]$ is allocated in the first level, for any $y \in Y_2$, $(\hat{x}_2, y)$ is connected with $(\hat{x}_1, y)$ by a between-level path. Furthermore, for any $y'\neq y \in Y$, $(\hat{x}_1, y)$ is connected with $(\hat{x}_1, y')$ by a between-level path. Then, the nodes in $[\hat{x}_1] \cup ([\hat{x}_2] \cap Y_2)$ are connected with each other via a between-level path. For $[\hat{x}_3]$, the nodes in $[\hat{x}_3] \cap Y_1$ and $[\hat{x}_3] \cap Y_2$ are connected to the corresponding nodes in $[\hat{x}_1] \cap Y_1$ and $[\hat{x}_2] \cap Y_2$ via a between-level path, respectively. Additionally, the nodes in $[\hat{x}_2] \cap Y_1$ are connected to the nodes in $[\hat{x}_3] \cap Y_1$ via a between-level path. Therefore, the sub-2-semi-chain restricted to $[\hat{x}_1] \cup [\hat{x}_2] \cup [\hat{x}_3]$ is strongly connected.
    
    Moreover, in both levels of this sub-2-semi-chain, for any $y \in Y$, there is $\hat{\theta}_y \in [y]$. Therefore, no matter whether divide the rest $[x]$-equivalent classes or not and what division patterns, the 2-semi-chain are strongly connected.
\end{proof}

\begin{proof}[Proof of Proposition~\ref{prp:dp_2dim}]
    Assume the sequence of division patterns contains $m$ divided $[x]$-equivalent class and $n_1 - m$ undivided $[x]$-equivalent class. Given $m \leq n_1$, the algorithm generates all sequences $\{d_{t_1}, \ldots, d_{t_m}\}$ with $1 \leq t_1 \leq \cdots \leq t_m \leq 2^{n_2}-2$ (recall that $\{d_3, \ldots, d_{2^{n_2}}\}$ are the set of the ordered partitions of $Y$,  $d_{2^{n_2}-1} = (Y,\emptyset)$ and $d_{2^{n_2}} = (\emptyset, Y)$). By Lemma~\ref{lemma:two_division}, the first part of Algorithm~\ref{algorithm:2_level_structure} generates all strongly connected sequences of division patterns in Category~I. By Lemma~\ref{lemma:feasible_same_level}, the second part generates those in Category~II. Finally, by Lemma~\ref{lemma:2_level_undivision}, the third part generates those in Category~III. Moreover, since the algorithm enumerates all sequences $(d_{t_1}, \ldots, d_{t_m})$ with $1 \leq t_1 \leq \cdots \leq t_m \leq 2^{n_2}-2$ without repetition, no duplicated sequence of division patterns is produced.

\end{proof}

\end{document}